\documentclass[journal]{IEEEtran}
\usepackage[caption=false,font=footnotesize]{subfig}
\usepackage{algorithm}
\usepackage{algpseudocode}
\usepackage{amsmath}
\usepackage{amsfonts}
\usepackage{amssymb}
\usepackage{amsthm}
\usepackage{mathtools}
\usepackage{color}
\usepackage[detect-weight=true]{siunitx}
\usepackage{nicefrac}
\usepackage{bbm}
\usepackage{booktabs}
\usepackage{pgfplots}
\usetikzlibrary{patterns}
\usepackage{layouts}
\usepackage{cite}
\usepackage{bm}
\usepackage[hidelinks]{hyperref}%

\pdfoutput=1%

\definecolor{lines-1}{RGB}{228,26,28}
\definecolor{lines-2}{RGB}{55,126,184}
\definecolor{lines-3}{RGB}{77,175,74}
\definecolor{lines-4}{RGB}{152,78,163}
\definecolor{lines-5}{RGB}{255,127,0}
\definecolor{lines-6}{RGB}{255,255,51}
\definecolor{lines-7}{RGB}{166,86,40}
\definecolor{lines-8}{RGB}{247,129,191}
\definecolor{lines-9}{RGB}{153,153,153}

\definecolor{verylightgrey}{RGB}{200,200,200}

\pgfplotscreateplotcyclelist{myCycleList}{
	color=lines-1,thick,mark=o,mark size=3\\%
	color=lines-2,thick,mark=star,mark size=3\\%
	color=lines-3,thick,mark=square,mark size=3\\%
	color=lines-4,thick,mark=diamond,mark size=3\\%
	color=lines-5,thick,mark=triangle,mark size=3\\%
}

\pgfplotsset{
	width =\columnwidth, 
	height=.8\columnwidth,
}

\DeclareMathOperator{\E}{\operatorname{\mathbb{E}}}

\newtheorem{definition}{Definition}
\newtheorem{lemma}{Lemma}
\newtheorem{theorem}{Theorem}

\algdef{SE}[DOWHILE]{Do}{doWhile}{\algorithmicdo}[1]{\algorithmicwhile\ #1}%

\newcommand{\pushright}[1]{\ifmeasuring@#1\else\omit\hfill$\displaystyle#1$\fi\ignorespaces}
\newcommand{\pushleft}[1]{\ifmeasuring@#1\else\omit$\displaystyle#1$\hfill\fi\ignorespaces}
\makeatother

\begin{document}
\title{Direct Localization for Massive MIMO}

\author{Nil~Garcia, Henk~Wymeersch, Erik~G.~Larsson, Alexander~M.~Haimovich, and Martial~Coulon\thanks{N.~Garcia and H.~Wymeersch are with the Department of Signals and Systems, Chalmers University of Technology, Gothenburg, Sweden. E.~G.~Larsson is with  the Division of Communication Systems, Department of Electrical Engineering (ISY), Link\"{o}ping University, Link\"{o}ping, Sweden. A.~M.~Haimovich is with  the Center for  Wireless  Communications   and  Signal  Processing  Research, ECE  Department,  New  Jersey  Institute  of  Technology  (NJIT), Newark, USA.  M.~Coulon is with  the IRIT/INP-ENSEEIHT, University of Toulouse, Toulouse, France.  This research was supported in part, by the European Research Council, under Grant No.~258418 (COOPNET),  the  EU HIGHTS project (High precision positioning for cooperative ITS applications) MG-3.5a-2014-636537. }}

\maketitle


\begin{abstract}
Large-scale MIMO systems are well known for their advantages in communications, but they also have the potential for providing very accurate localization thanks to their high angular resolution. A difficult problem arising indoors and outdoors is localizing users over multipath channels.
Localization based on angle of arrival (AOA) generally involves a two-step procedure, where signals are first processed to obtain a user's AOA at different base stations, followed by triangulation to determine the user's position. In the presence of multipath, the performance of these methods is greatly degraded due to the inability to correctly detect and/or estimate the AOA of the line-of-sight (LOS) paths.
To counter the limitations of this two-step procedure which is inherently sub-optimal, we propose a direct localization approach in which the position of a user is localized by jointly processing the observations obtained at distributed massive MIMO base stations.
Our approach is based on a novel compressed sensing framework that exploits channel properties to distinguish LOS from non-LOS signal paths, and leads to improved performance results compared to previous existing methods.
\end{abstract}

\section{Introduction}

\IEEEPARstart{M}{assive} MIMO, a leading 5G technology \cite{rappaport2013millimeter}, relies on the use of a large number of antennas at the base station. It has many advantages in cellular communications, including increased spectral efficiency, high directivity, and low complexity \cite{larsson2014massive,chin2014emerging}. While research on massive MIMO has focused mainly on communications, it is also an enabler for high-accuracy localization \cite{guerra2015position}. For instance, a finger-printing localization solution is proposed in \cite{savic2015fingerprinting} for locating  multiple users by means of  distributed massive MIMO. A personal mobile radar with millimeter-wave massive arrays is proposed in \cite{guidi2016personal} and used for simultaneous localization and mapping (SLAM) in \cite{guidi2014millimeter}. 

MIMO localization has received significant treatment in the technical literature, generally harnessing angle-of-arrival (AOA) estimation. Typically, a source emits a signal, and then in a \emph{two-step localization} approach, the AOAs are measured at all base stations, and then, the source's location is found by triangulation. In benign open-air applications, where such methods are referred to as bearings-only target localization (BOTL) good performance can be observed \cite{gavish1992performance,kaplan2001maximum,douganccay2004passive}.  However, in dense mutipath environments, such as urban areas or inside buildings, the AOA estimates are biased in general. For that reason, 
pure AOA-based techniques \cite{azzouzi2011new} have not been very popular in harsh multipath environments environments, due to large localization errors \cite{savic2015fingerprinting}. Massive arrays offer the possibility of precisely estimating the AOAs of the individual multipath components thanks to their high angular resolution. Nonetheless, measuring multiple AOAs at each base station requires  identification of the AOA of the line-of-sight (LOS) paths. A possible strategy is to select the strongest arrival as LOS \cite{klukas1998line}. However, the LOS path may be damped or obstructed as it is often the case indoors \cite{spencer2000modeling,sen2013avoiding}. Another option is to combine all estimated AOAs in a fusion center and perform data association, but this is an NP-hard problem to which no optimal solution exists \cite{pattipati1992new}.

\begin{figure}
	\begin{tikzpicture}
	\begin{axis}[
		xlabel={x-axis [\si{\meter}]},
		ylabel={y-axis [\si{\meter}]},
		xmin=-50, xmax=50,
		ymin=-50, ymax=50,
		axis line style = {opacity=0},
		legend style={font=\scriptsize,at={(0.5,0.97)},anchor=north},
		legend cell align=left,
		height=.9\columnwidth,
	]
	\addplot graphics[xmin=-50,ymin=-50,xmax=50,ymax=50] {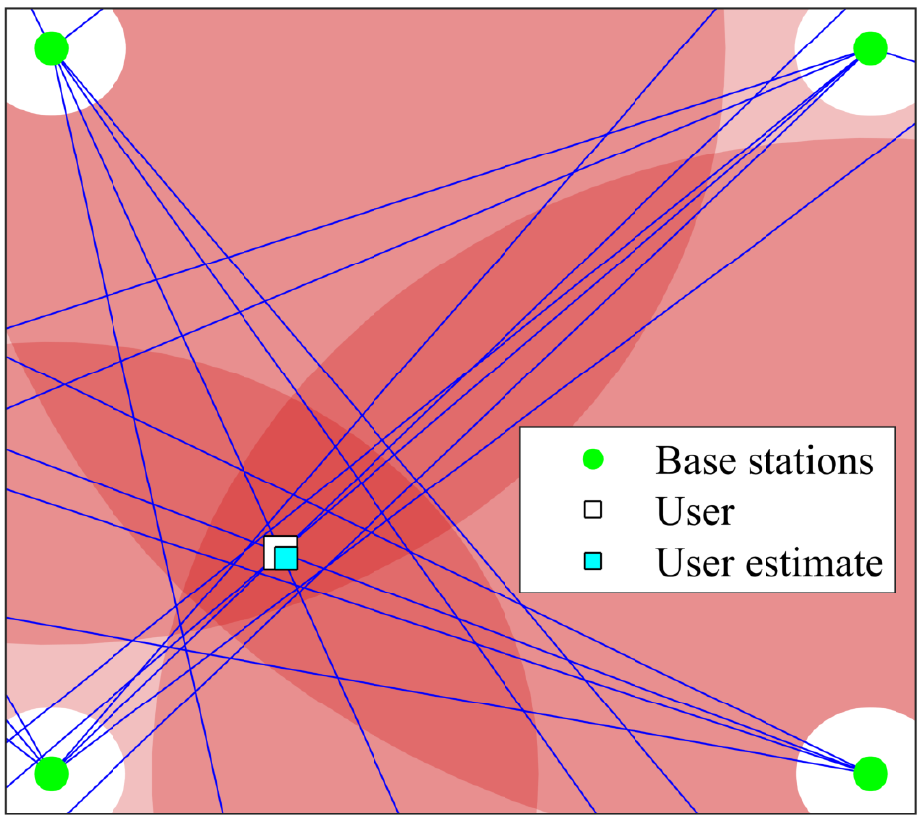};
	\end{axis}
	\end{tikzpicture}
	\caption{Example scenario. The locations of the base stations are indicated by the green dots and the true location of the source by the white square. The white circle around each base station is the region that is not in the far field according to the Fraunhofer formula. The larger red circles centered around each base station have radius proportional to the estimated TOAs. The intersection of these circles defines a feasible area (dark shaded red) where the source must be positioned. The blue lines represent the true bearing lines of the multipath arrivals at all base stations. Only the bearing lines of the multipath arrivals whose energy at each antenna of the pre-processed signals \eqref{eq:signal_model} is larger than the noise have been plotted.}
	\label{fig:scenario}
\end{figure}

An alternative way to tackle localization problems, is to use a \emph{direct localization} approach \cite{weiss2004direct}. Contrary to traditional techniques, the location of the source is estimated directly from the data, without estimating intermediate parameters, such as the AOAs of the LOS paths. The concept of direct localization was first introduced in \cite{wax1982location,wax1983optimum}, and later applied to AOA-based localization \cite{wax1985decentralized} and, more recently, to hybrid AOA--TOA (time-of-arrival) localization  \cite{weiss2004direct,weiss2005direct}. 
However, all these methods were designed for pure LOS environments. Some direct localization techniques \cite{chen2002maximum,bialer2013maximum} targeted to multipath scenarios exist in the literature, but they are not tailored to AOA information and massive arrays.
A requirement of direct localization is that the signals, or a function of them, are sent to a fusion center that estimates the source's locations. In general, it is easier to achieve such a topology indoors as the distances are smaller. In the case of cellular networks, cloud radio access networks (C-RAN) \cite{mobile2011cran,wu2012green} may provide the required infrastructure. C-RAN is a novel architecture for wireless cellular systems whereby the base stations relay the received signals to a central unit which performs all the baseband processing.

In this paper, we propose Direct Source Localization (DiSouL), a novel localization technique that jointly processes the snapshots of data acquired at each base station in order to directly estimate the location of the source. This technique, first introduced in \cite{garcia2014direct}, is 
based on compressive sensing at a fusion center and exploits the fact that LOS components must originate from a common location whereas NLOS components have arbitrary AOAs. 
In addition, two variations are presented to lower the computational burden and increase the precision of the position estimate. The first variation uses coarse TOA estimates at each base station to narrow the search area, while the second variation relies on a  grid refinement procedure \cite{malioutov2005sparse}. Finally, to validate the theory, numerous numerical results are provided.

Notation: $\lVert\cdot\rVert_1$, $\lVert\cdot\rVert_2$ (or $\lVert\cdot\rVert$) and $\lVert\cdot\rVert_{2,1}$ denote the $\ell_1$-norm, $\ell_2$-norm and $\ell_{2,1}$-norm, respectively, and $\lVert\cdot\rVert_0$ is the pseudo-$\ell_0$-norm which counts the number of non-zero elements.

\section{System Model}

We consider a two-dimensional scenario with one user (the source, located at $\mathbf{p}=[p^x,p^y]^{\mathrm{T}}$ in an area $\mathcal{R}$) and $L$  massive MIMO base stations with $S_l$ antennas each (located at $\tilde{\mathbf{p}}_{l}=[\tilde{p}^{x}_{l},\tilde{p}^{y}_{l}]^{\mathrm{T}}$, defined as the center of gravity of the associated phased array and assumed to be in the far field with respect to the source). We denote by $\mathbf{a}_l\left(\theta\right)$ the array response vector at base station $l$ for a ray impinging with angle $\theta$.

The source broadcasts a known signal $s(t)$ with half-power bandwidth $B$, which propagates through the multipath environment, resulting in a received signal at base station $l$ given by 
\begin{equation} \label{eq:signal_model_preMF}
	\mathbf{z}_l(t) = \mathbf{z}_l^\text{LOS}(t)+\mathbf{z}_l^\text{NLOS}(t) \qquad 0\leq t <T_\text{obs},
\end{equation}
where
\begin{align}
	\mathbf{z}_l^\text{LOS}(t) &=\alpha_l\, \mathbf{a}_l\left(\theta_l(\mathbf{p})\right) s(t-\tau_l(\mathbf{p})) \label{eq:signal_model_preMF:LOS}\\
	\mathbf{z}_l^\text{NLOS}(t) &=
	\sum_{m=1}^{P_l} \alpha_l^m\, \mathbf{a}_l(\theta_l^m) s(t-\tau_l^m)
	+\mathbf{n}_l(t), \label{eq:signal_model_NLOS}
\end{align}
in which $T_\text{obs}$ is the observation time, each component of $\mathbf{n}_l(t)$ is white Gaussian noise with spectral density $\sigma^2$,
$\alpha_l$  is an unknown complex scalar,  $\theta_l$ and $\tau_l$ are the angle of arrival (AOA) and time of arrival (TOA), all related to the line of sight (LOS) path, while $\alpha_l^m$, $\theta_l^m$, and $\tau_l^m$ are the channel gain, AOA, and TOA of the $m$-th NLOS component, for the $P_l$ NLOS path.  All these parameters are unknown. The signal is narrowband with respect to the arrays, i.e., at each array the amplitudes $\{\alpha\}$ do not change across antennas. The LOS parameters $\tau_l(\mathbf{p})$ and $\theta_l(\mathbf{p})$ are related to the source position through
\begin{align}
\tau_l(\mathbf{p}) & = \Vert \mathbf{p} - \tilde{\mathbf{p}}_{l} \Vert /c \\
\theta_l(\mathbf{p}) & =\arctan \left(\frac{p^{y}-\tilde{p}^{y}_{l}}{p^{x}-\tilde{p}^{x}_{l}}\right)
	+\pi\cdot\mathbbm{1}\left(p^{x}<\tilde{p}^{x}_{l}\right),
\end{align}
where $c$ is the speed of light, while the range of the arctangent function is ${-\pi}/{2}\leq\arctan(x)<+{\pi}/{2}$, the angle is computed with respect to the $x$-axis and anticlockwise, and $\mathbbm{1}(\mathsf{P})$ is one if the logical expression $\mathsf{P}$ is true\footnote{The summand $\pi\cdot\mathbbm{1}\left(p^{x}<\tilde{p}^{x}_{l}\right)$ is added for resolving the ambiguity caused by the fact that $\arctan({y}/{x})=\arctan({-y}/{-x})$.}.

We generate a discrete-time observation, by applying a matched filter\footnote{Performing the matched filtering requires perfect knowledge of the pulse shape $s(t)$. In practice, if the antennas and hardware have an entirely all-pass (frequency-flat) frequency response, then the signal-to-noise-ratio will decrease but the number of multipath components will remain the same.} and sampling at time $t_l$, leading to 
\begin{equation} \label{eq:signal_model}
	\begin{split}
		\bar{\mathbf{z}}_l &= \int_{0}^{T_\text{obs}} s^*(t-t_l)  \mathbf{z}_{l}(t) \, \textrm{d}t \\ 
		&=\bar{\alpha}_l\, \mathbf{a}_l\left(\theta_l(\mathbf{p})\right)
		+\sum_{m=1}^{P_l} \bar{\alpha}_l^m\, \mathbf{a}_l\left(\theta_l^m\right)
		+\bar{\mathbf{n}}_l
	\end{split}
\end{equation}
 where $\bar{\alpha}_l = r_s(t_l-\tau_l(\mathbf{p})){\alpha}_l$, $\bar{\alpha}_l^m = r_s(t_l-\tau_l^m){\alpha}_l^m$ and $r_s(t) = \int_{0}^{T_\text{obs}} s^*(t-t_l)s(t) \,\textrm{d}t$ is the emitted signal autocorrelation. The pulse energy is normalized to one, i.e., $r_s(0)=1$. Therefore, $\bar{\mathbf{n}}_l \sim \mathcal{N}(\mathbf{0},\sigma^2 \mathbf{I})$. The signals in \eqref{eq:signal_model} are the input of the proposed method. In order to ensure that $\bar{\alpha}_l \neq 0$, the signals must be sampled at a time where the energy of the LOS pulse is not zero (i.e., while $r_s(t_l-\tau_l(\mathbf{p}))\neq0$). In addition to the sampling times, we will also compute an upper bound on the TOA of the LOS paths at each base station that will enhance the proposed method, which for brevity are simply called TOA estimates.   In a nutshell, the objective of our work is to determine the sampling times and the TOA estimates from $\mathbf{z}_l(t) \}^{L}_{l=1}$, and then, determine $\mathbf{p}$ from $\{ \bar{\mathbf{z}}_l \}^{L}_{l=1}$.

\section{Proposed Method}

\subsection{Principle} \label{sec:principle}

The proposed method exploits the high angular resolution of massive arrays, enabling detection and estimation of the AOAs of the distinct multipath arrivals. The TOA estimates are not used for precise localization, but rather as constraints, limiting the source location in a convex set (See Fig.~\ref{fig:scenario} for a visual example). By processing the snapshots $\bar{\mathbf{z}}_l$ at all base stations jointly, we are able to separate the LOS paths from the NLOS paths. 
Roughly speaking, the procedure of DiSouL is as follows. First, determine a coarse, positively biased estimate of the TOA at each base station.
The TOA estimates define a convex set containing the source. 
Second, using the signal model \eqref{eq:signal_model}, we formulate a convex optimization problem, directly providing an estimate of $\mathbf{p}$. 
In contrast to  indirect approaches, where AOAs are estimated first and the source location is determined afterwards, we do not have to  deal with an NP-hard data-association problem. We bypass this problem, as only position $\mathbf{p}$ is estimated and not AOAs of the LOS paths. While the TOA estimates need to be an upper bound on LOS TOAs, i.e., positively biased, the sampling should happen at an instant where the energy of the LOS arrival is maximized with respect to the energy from the NLOS arrivals. Thus, in general the sampling times will be smaller than the TOA estimates.

 \subsection{TOA Estimation and Sampling} \label{sec:timing}

 The proposed method requires computing the instants $\{t_l\}_{l=1}^L$ at which the outputs of the matched filters are sampled \eqref{eq:signal_model}, along with the TOA estimates $\{\hat{\tau}_l\}_{l=1}^L$. 
 Here, a delay estimation technique will be described for computing these time measurements. Nonetheless, other delay estimation techniques are possible and may lead to better results.
 The desired properties of the sampling times and TOA estimates are different, and consequently, in this work they are computed differently. 
 Because NLOS components are treated as interference, the sampling time shall be picked such that the ratio between the energy of the LOS component and the aggregated energy of the NLOS components is maximized. 
 
 \subsubsection{TOA Estimation} \label{sec:TOA}
 
We rely on a generalization of the threshold MF \cite{dardari2008threshold} for multiple antennas. Let $z^\text{NC}_{l}(\tau)$ be the non-coherent aggregation of the observed signals at all antennas after matched filtering:
 \begin{equation}
 z^\text{NC}_{l}(\tau) = \Big\| \int_{0}^{T_\text{obs}} s^*(t-\tau)  \mathbf{z}_{l}(t)\, \textrm{d}t \Big\|^2.
 \end{equation}
 The TOAs $\{\hat{\tau}_l\}_{l=1}^L$ are estimated by selecting the first peak\footnote{A \emph{peak} is a local maximum of $z^\text{NC}_{l}(\tau)$.} that exceeds a certain threshold:
 \begin{equation} \label{eq:threshold_MF}
 \hat{\tau}_l=\operatorname{find-1st-peak} \left\{z^\text{NC}_{l}(\tau) : z^\text{NC}_{l}(\tau) \geq \eta\right\}.
 \end{equation}
 In practice, values of $z^\text{NC}_{l}(\tau)$ may only be available at discrete instants. In such case the location of the peak may be obtained by parabolic fitting \cite{cespedes1995methods}.
 Following \cite{dardari2008threshold}, the threshold is selected so that the probability of early false alarm is very low. An early false alarm event is defined as detecting a peak due to noise before the true TOA of the LOS path. The probability of early false alarm can be approximated by \cite{dardari2008threshold}  
 \begin{equation} \label{eq:threshold_estimate}
 P_{\textrm{FA}} \approx 1 +
 \frac{\left(1-\exp\left(-\frac{\eta}{S_l \,\sigma^2}\right)\right)^{{T_\text{obs}}/{T_\text{corr}}}-1}
 {\exp\left(-\frac{\eta}{S_l \,\sigma^2}\right)^{{T_\text{obs}}/{T_\text{corr}}}}.
 \end{equation}
 The value of threshold $\eta$ resulting in the desired $P_{\textrm{FA}}$ can be found by performing a one-dimensional search of \eqref{eq:threshold_estimate}. 
 Generally speaking, for many types of waveforms, the correlation time is well approximated by the inverse of the bandwidth: $T_\text{corr}={1}/{B}$.

 \subsubsection{Sampling Time}
 
 The received signals at all antennas of base station $l$ \eqref{eq:signal_model} are sampled at time $t_l$. Contrary to the estimation of the TOAs, the goal is to sample at an instant where there is as little as possible NLOS interference and as much as possible energy from the LOS component. Hence, we propose to use the same threshold matched filter for TOA estimation, but instead of selecting the time of the first peak, we select the instant when the received signal crosses the threshold for the first time, i.e.,
 \begin{equation} \label{eq:threshold_MF_1stCross}
 t_l = \min \left\{\tau : z^\text{NC}_{l}(\tau) \geq \eta\right\}.
 \end{equation}

\subsection{Localization}

To solve the localization problem, we rely on tools from compressive sensing. Specifically, we propose a grid-based approximate solution to the the problem of localizing a source on a continuous map which exploits the notion of sparsity and row-sparsity \cite{tropp2006algorithms}. To this end,
first, we introduce a uniform grid of $Q$ locations  
\begin{equation} \label{eq:grid_locs}
	\mathcal{L}=\left\{
	\boldsymbol{\bm{\pi}}_1,\ldots,\boldsymbol{\bm{\pi}}_Q
	\right\} \subset \mathcal{R},
\end{equation}
and a uniform grid of $M_l$ angles for each base station array
\begin{equation} \label{eq:grid_angles}
	\mathcal{A}_l =\left\{
	\vartheta_1,\ldots,\vartheta_{M_l}
	\right\} \subset [0,2\pi).
\end{equation}
The main assumption here is that the source is positioned on a grid location, and that the AOAs of the NLOS paths are also in the grid of angles.
Let $\mathbf{X} \in \mathbb{C}^{Q \times L}$ be a matrix whose entry on row $q$, column $l$ is denoted by $x_{ql}$ and represents the complex gain of a LOS path from grid location $\boldsymbol{\bm{\pi}}_q$ to base station $l$. Let $y_{ml}$ be the complex gain of a NLOS path arriving at the $l$-th base station with angle $\vartheta_m$. 
Then, by definition, only one row in $\mathbf{X}$ is different from zero, and $y_{ml}\neq 0$ only if $\vartheta_m$ is equal to the AOA of a NLOS path at base station $l$. Thus, if the grids are dense enough, $\mathbf{X}$ is row-sparse and $\mathbf{y}_l$ is sparse for all $l$. It is well known in the compressive sensing literature \cite{jacob2009group}, that row sparsity can be induced by minimizing the $\ell_{2,1}$-norm, i.e., $\Vert  \mathbf{X} \Vert_{2,1} = \sum_{q=1}^{Q} \sqrt{\sum_{l=1}^{L} \left| x_{ql} \right|^2}$, and that sparsity can be induced by minimizing the $\ell_1$-norm, i.e., $\Vert \mathbf{y}_l \Vert_1=  \sum_{m=1}^{M_l} \left| y_{ml} \right|$ where $\mathbf{y}_l=[y_{1l},\ldots,y_{M_ll}]^\mathrm{T}$. Thus, with all this in mind, we propose to solve the following optimization problem
\begin{subequations} \label{opt:problem}
	\begin{align}
		\label{opt:problem:objective}
		\min_{\mathbf{X},\mathbf{y}_l} \quad&
		w \Vert  \mathbf{X} \Vert_{2,1} + \sum_{l=1}^L \Vert \mathbf{y}_l \Vert_1 \\
		\label{opt:problem:constraint}
		\text{s.t.}
		\quad&
		\sum_{l=1}^{L}
		\left\|
		\bar{\mathbf{z}}_l - \hat{\mathbf{z}}_l
		\right\|^2
		\leq \epsilon
		\\ \quad&
		\hat{\mathbf{z}}_l =		
		\sum_{q=1}^{Q} x_{ql} \mathbf{a}_l\left(\theta_l(\boldsymbol{\bm{\pi}}_q)\right)
		+\sum_{m=1}^{M} y_{ml} \mathbf{a}_l\left(\vartheta_m\right), \forall l. \label{opt:problem:signal_estimated}
	\end{align}
\end{subequations}
This is a second-order cone program (SOCP) for which very efficient solvers exist. Intuitively, it seeks the sparsest number of source locations and NLOS paths that can explain the observations \eqref{eq:signal_model}.  The vector $\hat{\mathbf{z}}_l$ as defined  in \eqref{opt:problem:signal_estimated} is a reconstruction of $\bar{\mathbf{z}}_l$ for a given choice of $\mathbf{X}$ and $\left\{\mathbf{y}_l\right\}_{l=1}^{M_l}$. 
The parameter $\epsilon$ establishes the maximum allowed mismatch between the observations and the reconstruction. The parameter $w$ ensures that not all signal energy is assigned to only LOS or NLOS components. Suitable choices for $\epsilon$ and $w$ will be proposed below.

\subsubsection*{Remark} While this technique searches a source on a plane, it can be generalized to a three-dimensional search, at a cost of computational complexity. It is also possible that the technique may improve its robustness against multipath because in two dimensions two distinct NLOS bearing lines always intersect but in three dimensions they generally do not.

\section{Parameter Selection}
In this section, we motivate the choice for $\epsilon$ and $w$. The choices do not guarantee recovery of the correct position and are derived under simplified assumptions. 

\subsection{Setting the Parameter $\epsilon$}

The parameter $\epsilon$ in \eqref{opt:problem} defines the allowed mismatch between the observations and the reconstruction. Typically, $\epsilon$ is a bound on the noise. Since the noise is Gaussian, it is unbounded, and instead $\epsilon$ is chosen so that the received signals in absence of noise are part of the feasible set with high probability, i.e.,
\begin{equation}
	\operatorname{Prob}
	\left(\sum_{l=1}^{L}\left\|\bar{\mathbf{z}}_l-\hat{\mathbf{z}}_l\right\|^2 \leq \epsilon  \right) = \gamma,
\end{equation}
where $\gamma$ is, for example, 0.99. Substituting the observations by their expression \eqref{eq:signal_model}, we obtain an expression that only depends on the noise
\begin{equation}
	\operatorname{Prob}\left(\sum_{l=1}^{L}\left\|\bar{\mathbf{n}}_l\right\|\leq\epsilon\right) = \gamma.
\end{equation}
Because $\bar{\mathbf{n}}_l$ are random white Gaussian vectors of length $S_l$, it follows that the
error normalized by the noise variance $2\sigma^{-2}\sum_{l=1}^{L}\|\bar{\mathbf{n}}_l\|$ is a Chi-squared random variable with $2\sum_{l=1}^{L}S_l$ degrees of freedom. Let $F_k(x,k)$ be the cumulative distribution function of the Chi-squared distribution with $k$ degrees of freedom evaluated at $x$ and $F^{-1}(y,k)$ its inverse function evaluated at $y$. Then, the value of $\epsilon$ can be computed as
\begin{equation} \label{eq:param_epsilon}
	\epsilon = \frac{\sigma^2}{2}F^{-1}\left(\gamma,2\sum_{l=1}^{L}S_l\right).
\end{equation}

In low SNR conditions, it is possible that the aggregated energy of all snapshots is not larger than $\epsilon$, i.e.,
\begin{equation}
	\sum_{l=1}^{L}\left\|\bar{\mathbf{z}}_l\right\|^2 \leq \epsilon,
\end{equation}
making problem \eqref{opt:problem} have the trivial all-zeros solution, thus, failing to estimate the location of the source. In cases we propose to simply look up the location whose LOS components correlate the most with the snapshots:
\begin{equation} \label{eq:ML_LOS}
	\hat{\mathbf{p}} = \arg \max_{\boldsymbol{\bm{\pi}}\in\mathcal{L}} \sum_{l=1}^{L}\frac{\left|\mathbf{a}_l^\textrm{H}\left(\theta_l(\boldsymbol{\bm{\pi}})\right)\bar{\mathbf{z}}_l\right|^2}{\left\|\mathbf{a}_l\left(\theta_l(\boldsymbol{\bm{\pi}})\right)\right\|^2}.
\end{equation}

\subsection{Setting the Parameter $w$}
In order to obtain an expression for $w$, we will not  prove that the AOAs are correctly recovered by solving \eqref{opt:problem}, but rather that, under proper selection of $w$, if the AOAs are correctly recovered, then they can also be correctly identified as either LOS or NLOS. 
The key property that will dictate the value of $w$, and in turn estimate the correct source location is based on the following definition.
\begin{definition}[consistent location] \label{def:consistency}
	A location $\boldsymbol{\bm{\pi}}$ is consistent with $L$ paths, if the AOAs of the direct paths between $\bm{\pi}$ and the base stations are true AOAs, i.e.,
	\begin{equation} \label{eq:consistency}
		\theta_{l}(\boldsymbol{\bm{\pi}}) \in \Theta_l \quad\text{for }l=1,\ldots,L,
	\end{equation}
	where $\Theta_l$ is the set of true AOAs at base station $l$
	\begin{equation} \label{eq:setTrueAOAs}
	\Theta_l=
	\left\{\theta_l(\mathbf{p}),\theta_l^1,\ldots,\theta_l^{P_l}\right\}.
	\end{equation}
\end{definition}
By definition, the true source location $\mathbf{p}$ is consistent with the $L$ paths because the LOS components travel in a straight line. To find a criterion for the weight, we restrict ourselves to a simplified version of the problem and then later evaluate the criterion in a more realistic setting. 
Our analysis on the weight criterion is limited through the three following assumptions.
\begin{description}
	\item[A1)] Besides the source location $\mathbf{p}$, no other location is consistent with $L$ paths. \label{assum:consistency}
	\item[A2)] The grids $\mathcal{L}$ and $\{\mathcal{A}_l\}_{l=1}^L$ are sufficiently dense to contain the source location $\mathbf{p}$ and all AOAs, respectively. 
	\item[A3)] \label{assum:correctAOAs} Denoting by $\hat{\Theta}_l$ the estimated AOAs at base station $l$, i.e., 
	\begin{equation} \label{eq:setEstimatedAOAs}
		\hat{\Theta}_l=
		\left\{\vphantom{\vartheta_{ml} : y_{ml}\neq0,\vartheta\in[0,2\pi)}\theta_l(\boldsymbol{\bm{\pi}}_q):x_{ql}\neq0\right\}
		\cup 
		\left\{\vartheta_{ml} : y_{m l}\neq0\right\},
	\end{equation}
then $\hat{\Theta}_l={\Theta}_l,  \forall l$. In other words, the solution of \eqref{opt:problem} is able to recover the true AOAs. This assumption is reasonable in high SNR conditions. 
\end{description}


\begin{lemma} \label{lemma:proof:minLpaths} Assume A2) and A3). If $w > \sqrt{L-1}$,
	then any estimated location output by problem \eqref{opt:problem} is consistent with $L$ paths (in the sense of Definition~\ref{def:consistency}). 
\end{lemma}
\begin{proof}
See Appendix~\ref{proof:minLpaths}. 
\end{proof}

\begin{lemma} \label{lemma:min1source}
	Assume A2) and A3). If  $w < \sqrt{L}$,
	then problem \eqref{opt:problem} outputs at least one location (i.e., $\mathbf{X} \neq \mathbf{0}$). 
\end{lemma}
\begin{proof}
See Appendix~\ref{proof:min1source}.
\end{proof}

The two lemmas lead directly to the following theorem, which guarantees the correct recovery of the source location.
\begin{theorem} \label{thrm:weight}
	If Assumptions {A1)}, {A2)}, and {A3)} hold, then a sufficient condition for the correct recovery of the source location is 
	\begin{equation}
		\sqrt{L-1}<w<\sqrt{L}.
	\end{equation}
\end{theorem}

\begin{proof}
	If $w<\sqrt{L}$, by Lemma~\ref{lemma:min1source} al least one estimated location is output by problem~\eqref{opt:problem}.
	Moreover, if $w>\sqrt{L-1}$, by Lemma~\ref{lemma:proof:minLpaths} any estimated location is consistent with $L$ paths. However, according to Assumption {A1)}, only the location of the source is consistent
	with $L$ paths, thus completing the proof.
\end{proof}

\subsection{The Cases of Obstructed-LOS and Non-LOS}

In practice, LOS paths may be attenuated or blocked (leading to obstructed-line-of-sight (OLOS) or non-line-of-sight (NLOS), respectively). The proposed technique relies on the presence of the $L$ LOS paths for achieving high precision localization, and it may break down when the base stations are in NLOS. Similarly, OLOS base stations will receive LOS components that are attenuated and may pass undetected. If the weight in \eqref{opt:problem:objective} is chosen according to Lemma~\ref{lemma:proof:minLpaths}, any location estimate output by problem \eqref{opt:problem} must be consistent with $L$ paths. However, if one base station is in NLOS, then the source will only be consistent with $L-1$ paths, and therefore, the location of the source will not be a solution to \eqref{opt:problem}. Thus, adjusting the weight requires a priori knowledge of the number of LOS base stations. We can adjust the weight as follows.  Let  $L^*$ be the number of base stations in LOS  with the source, and let $\hat{L}$ be an estimate of $L^*$. Obviously, $L^*\leq L$. Furthermore, assume no other location besides the location of the source is consistent with $L^*$ paths. We start by assuming that all base stations are in LOS, i.e., $\hat{L}=L$, set the weight according to Theorem~\ref{thrm:weight}, and solve problem \eqref{opt:problem}. According to Lemma~\ref{lemma:proof:minLpaths}, the location of the source will be estimated only if it is consistent with $\hat{L}$ paths. If $L^*< \hat{L}$, the solver will return $\mathbf{X}=\mathbf{0}$. When this event is detected, $\hat{L}$ can be reduced and \eqref{opt:problem} solved again. This procedure can be repeated as shown in Algorithm~\ref{alg:NLOS_BSs}. Note that if $\mathbf{X}\neq\mathbf{0}$, the location with strongest gains is returned as the estimate $\hat{\mathbf{p}}$ (see lines 10--11). 

Based on the choices for $\epsilon$ and $w$, Algorithm~\ref{alg:NLOS_BSs} summarizes the proposed solution strategy.

\begin{algorithm} \caption{Direct Localization} \label{alg:NLOS_BSs}
	\begin{algorithmic}[1]
		
		
		
		
		\State set $\hat{L}=L$ and $\hat{\mathbf{p}}=\emptyset$
		
		\State set $\epsilon$ according to \eqref{eq:param_epsilon}
		
		\If{$\sum_{l=1}^{L}\left\|\bar{\mathbf{z}}_l\right\|^2 > \epsilon$}
		
			\While{$\hat{\mathbf{p}}=\emptyset$ and $\hat{L}>1$}
			
			\State set $w=\sqrt{\hat{L}-0.5}$
			
			\State solve  \eqref{opt:problem} to obtain $\mathbf{X}$ and  $\mathbf{y}_l, \forall l$ \label{line:solve}
			
			\If{$(\mathbf{X} \equiv \mathbf{0})$}
			
			\State $\hat{L} \leftarrow \hat{L}-1$
			
			\Else
			
						\State $\hat{q} = \operatorname{arg\,max}_q\, \left\|\mathbf{X}_{q,:}\right\|$
			
			\State $\hat{\mathbf{p}} = \boldsymbol{\bm{\pi}}_{\hat{q}}$
			
			\EndIf
			
			\EndWhile
			
		\Else
		
			\State estimate $\hat{\mathbf{p}}$ by \eqref{eq:ML_LOS}
		
		\EndIf

	\end{algorithmic}
\end{algorithm}

\section{Improving Computational Time and Precision}
The performance of the proposed method is determined by the quality of the sampling times and the density of the grids. The latter also relates directly to the computational complexity. In this section, we provide improvements to the basic algorithm with respect to these two 
aspects.

\subsection{TOA Assistance}\label{sec:TOAassist}

Assuming all TOA estimates are positively biased, then we can create a set 
\begin{align}
\mathcal{F} = \{ \bm{\pi} \in \mathbb{R}^2: \left\|\mathbf{\boldsymbol{\bm{\pi}}}-\mathbf{p}_{l}\right\| \leq c \, \hat{\tau}_l,  \forall l \}
\end{align}
and finally use $\mathcal{L} \cap \mathcal{F}$ instead of $\mathcal{L}$ in Algorithm~\ref{alg:NLOS_BSs}. In case $\mathcal{L} \cap \mathcal{F} = \emptyset$, a new grid can be generated in $\mathcal{F}$. 

In the unlikely event that not all TOA estimates are positively biased, it is possible that $\mathcal{F}=\emptyset$. In such a case, we expand $\mathcal{F}$ by increasing all TOA estimates by a constant value $v$ until $\mathcal{F}\neq \emptyset$. We have chosen $v=1/B$, where $B$ is the signal bandwidth, though the value of $v$ turns out to not be critical for the localization performance.

\subsection{Grid Refinement}

Dense grids of locations and angles are necessary to achieve fine resolution, but making the grids too dense results in large computation time. The computational complexity of solving \eqref{opt:problem} scales as $\mathcal{O}((QL+\sum_l M_l)^{3.5})$ \cite{lobo1998applications}, where $Q$ and $M$ are the number grid locations and angles respectively, and $L$ is the number of base stations. This motivates an adaptive grid-refinement strategy originally proposed in \cite{malioutov2005sparse}. The idea behind the grid refinement approach is to start with a coarse grid of locations and angles; subsequently, the grid is refined around the estimated locations and angles and the optimization problem \eqref{opt:problem} is solved again. This procedure can be repeated until a certain grid resolution has been achieved or a stopping criterion has been met. Thus, the benefits of grid refinement are two-fold: lower computational complexity and fine grid resolution.

In comparison to previous grid refinement approaches \cite{malioutov2005sparse,hyder2010direction}, ours is more complex due to the two different types of grids used to describe the observed data. At iteration $k$ of the grid refinement process, we will denote the position grid by $\mathcal{L}^{(k)}$ and the angle grid (for base station $l$) by $\mathcal{A}^{(k)}_l$.
At iteration $k=0$, the grids are uniform over $\mathcal{R}$ and $[0,2 \pi)$, respectively. The resolutions in $\mathcal{L}^{(0)}$ and $\mathcal{A}^{(0)}_l$ are set to ${\pi}_{\text{res}} \in \mathbb{R}_+$ and $\vartheta_\text{res}\in \mathbb{R}_+$, respectively. 
Define the following operators: 
\begin{align}
\mathrm{grid}(\hat{\Pi},\delta)=&\{ \bm{\pi} \in \mathbb{R}^2 : \bm{\pi} =\hat{\bm{\pi}} + [i~j]^{\mathrm{T}}  \delta,\\
& \hat{\bm{\pi}} \in \hat{\Pi} ,i,j \in \{-2,-1,0,1,2\}\} \nonumber \\
\mathrm{grid}(\hat{\Theta}_l,\delta)=& \{ \theta \in [0,2\pi): \theta = \hat{\theta} + i \delta,\\
&  \hat{\theta} \in \hat{\Theta}_l, i \in \{-2,-1,0,1,2\}\}. \nonumber
\end{align}
We can then set $\mathcal{L}^{(k)} = \mathrm{grid}(\hat{\Pi},{\pi}_{\text{res}}/2^{k})$ and $\mathcal{A}^{(k)}_l = \mathrm{grid}(\hat{\Theta}_l,\vartheta_\text{res}/{2^k}) \cup 	\{\left[\theta_l\left(\boldsymbol{\bm{\pi}}\right)\right]_{\vartheta_\text{res}/{2^k}}  :  \boldsymbol{\bm{\pi}}\in \hat{\Pi} \}$, where $[x]_{y} $ rounds $x$ to the nearest multiple of $y$. Each successive grid of locations and angles includes the estimated points and their neighboring points. In this case we have chosen to include twenty-four and four neighbour points for the position grid and angle grid, respectively, but other choices of neighbours are possible as well. In addition, the grid of angles also incorporates the angles related to the estimated locations. It has been empirically verified that this is necessary for the correct performance of this grid refinement approach.

Because at each step the previously estimated points are included in the next grid, the solution at step $k$ is a feasible solution at step $k+1$. This ensures that the optimum value of the optimization problem \eqref{opt:problem} cannot increase as iterations progress. Since the objective function is bounded from below by zero, by the monotone convergence theorem \cite{schechter1996handbook}, the grid refinement procedure must converge. In practice, the refinement process is halted when the progress between two consecutive steps is negligible. Denote as $f_\text{opt}^{(k)}$ the optimum value of problem \eqref{opt:problem} at step $k$, then the grid refinement is stopped at step $k$ if
\begin{equation} \label{eq:stop_criterion}
	\frac{|f_\text{opt}^{(k-1)}-f_\text{opt}^{(k)}|}{f_\text{opt}^{(k-1)}} < \beta,
\end{equation}
where $\beta$ is a small value, e.g., $\beta=10^{-3}$.
\begin{algorithm} \caption{Grid refinement} \label{alg:grid_refinement}
	\begin{algorithmic}[1]
		
		\State given a coarse grids of locations $\mathcal{L}^{(0)}$ and angles $\mathcal{A}_l^{(0)}, \forall l$
		
		

		\State set $k=0$
		
				\While{\eqref{eq:stop_criterion} not satisfied}

			\State solve  \eqref{opt:problem} with $\mathcal{L}=\mathcal{L}^{(k)}$ and $\mathcal{A}_l=\mathcal{A}_l^{(k)}$ 
			
			\State extract  locations $\hat{\Pi} =\{ \boldsymbol{\bm{\pi}}_q^{(k)} \in \mathcal{L}^{(k)} : \|\mathbf{x}_{q}^{(k)}\|\neq0  \}$
			
			\State extract  angles $\hat{\Theta}_{l} =\{ \vartheta_{ml}^{(k)} \in \mathcal{A}_l^{(k)} : y_{ml}^{(k)}\neq0 \}, \forall l$ 
			\State increase $k$			
			\State set $\mathcal{L}^{(k)}=\mathrm{grid}(\hat{\Pi},{\pi}_{\text{res}}/2^{k})$
			\State trim grid of locations $\mathcal{L}^{(k)}$ through TOA assistance
			\State set
			\begin{equation}
			\mathcal{A}^{(k)}_l = \mathrm{grid}(\hat{\Theta}_l,\vartheta_\text{res}/{2^k}) \cup 	\{\left[\theta_l\left(\boldsymbol{\bm{\pi}}\right)\right]_{\vartheta_\text{res}/{2^k}}  :  \boldsymbol{\bm{\pi}}\in \hat{\Pi} \}	 \nonumber					
			\end{equation}
		\EndWhile

	\end{algorithmic}
\end{algorithm}

\subsection{The DiSouL Algorithm}

The summary of the  DiSouL algorithm is now presented, comprising the basic Algorithm~\ref{alg:NLOS_BSs}, as well as the TOA assistance and grid refinement.

\begin{algorithm}[H] \caption{DiSouL} \label{alg:DiSouL}

	\begin{algorithmic}[1]
		
		\State set $\eta$ using \eqref{eq:threshold_estimate} for the desired $P_{\textrm{FA}}$
		
		\State estimate TOAs $\{\hat{\tau}_l\}_{l=1}^L$ using \eqref{eq:threshold_MF}
		
		\State create initial grid of locations $\mathcal{L}^{(0)}$ and angles $\mathcal{A}_{l}^{(0)}, \forall l$ 
		
		\State trim grid of locations through TOA assistance 
		
		\State compute sampling times $\{t_l\}_{l=1}^L$ using \eqref{eq:threshold_MF_1stCross}
		
		\State obtain the snapshots of data by applying MF and sampling at instants $\{t_l\}_{l=1}^L$ as in \eqref{eq:signal_model}
		
		\State estimate source location $\hat{\mathbf{p}}$ by Algorithm~\ref{alg:NLOS_BSs} where line~\ref{line:solve} is replaced with Algorithm~\ref{alg:grid_refinement}
		
	\end{algorithmic}
\end{algorithm}

\section{Numerical Results}

In this section, we illustrate the performance of the localization method and compare it to other existing techniques. Unless otherwise stated, all numerical examples are run using the following parameters. The source is positioned randomly within an area of size \SI{100x100}{\meter}. Four base stations are positioned at the corners; if the origin of the coordinate system is taken to be in the middle of the area, the base stations are at coordinates [\SI{45}{\meter},\SI{45}{\meter}], [\SI{45}{\meter},\SI{-45}{\meter}], [\SI{-45}{\meter},\SI{45}{\meter}] and [\SI{-45}{\meter},\SI{-45}{\meter}]. Every base station is equipped with a 100-antenna circular random array \cite{ochiai2005collaborative} of radius $5\lambda$. In a random circular array all antennas are placed uniformly at random inside a disk, which lies in the same plane as the search area. The carrier frequency is \SI{7}{\giga\hertz}. We opt for circular random arrays instead of uniform linear arrays (ULAs) because for the same number antennas and similar inter antenna spacing, their far field region starts at a much shorter distance (at \SI{8.6}{\meter} instead of  \SI{51.4}{\meter} for a ULA) \cite{std1979antennas}. 
The initial grid resolutions of DiSouL are ${\pi}_{\text{res}}=\SI{5}{\meter}$ and $\vartheta_\text{res}=\SI{5.71}{\degree}$.

\subsection{Validation of Theorem \ref{thrm:weight}}

\begin{figure}
	\centering
	\includegraphics[width=.48\columnwidth]{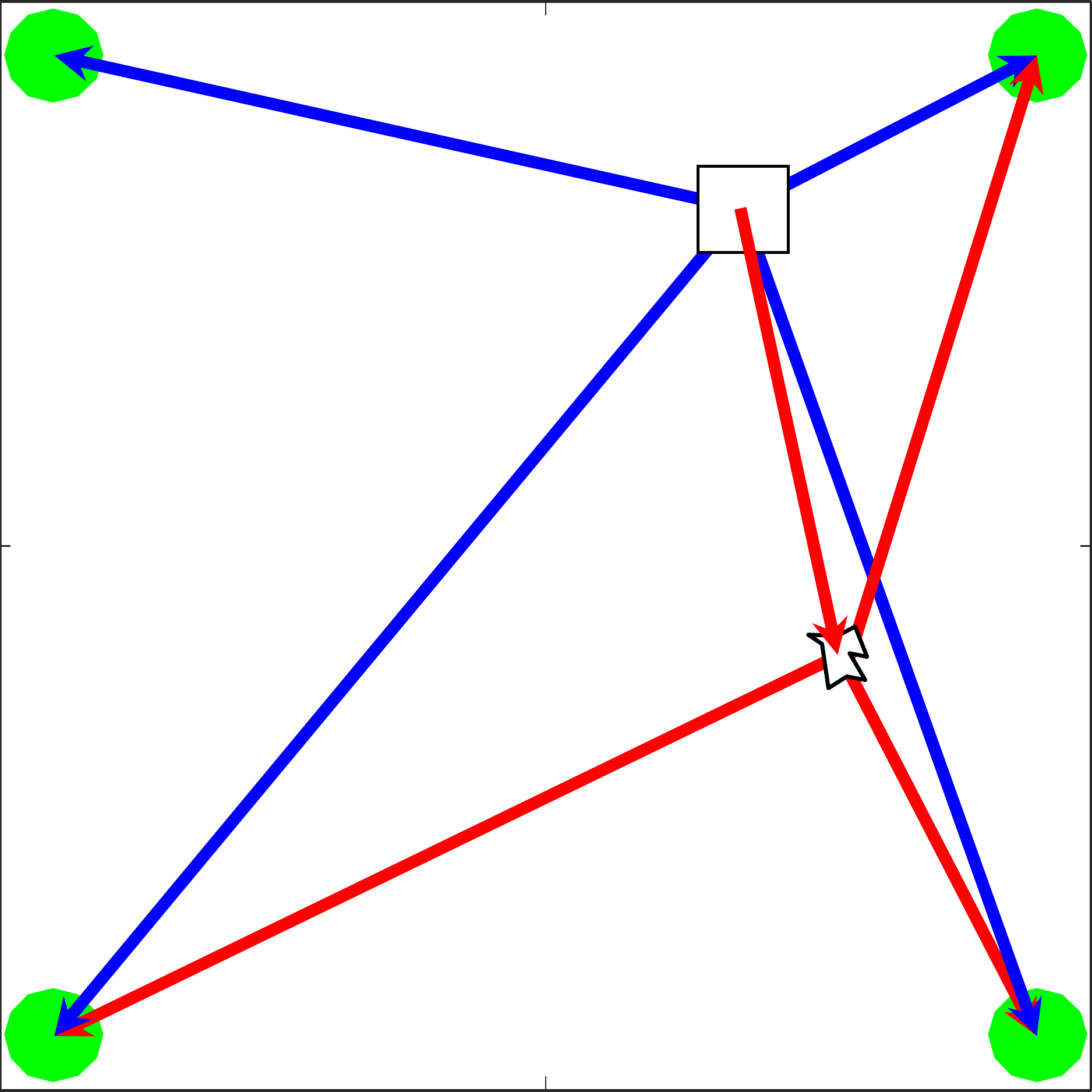}
	\hfill
	\includegraphics[width=.48\columnwidth]{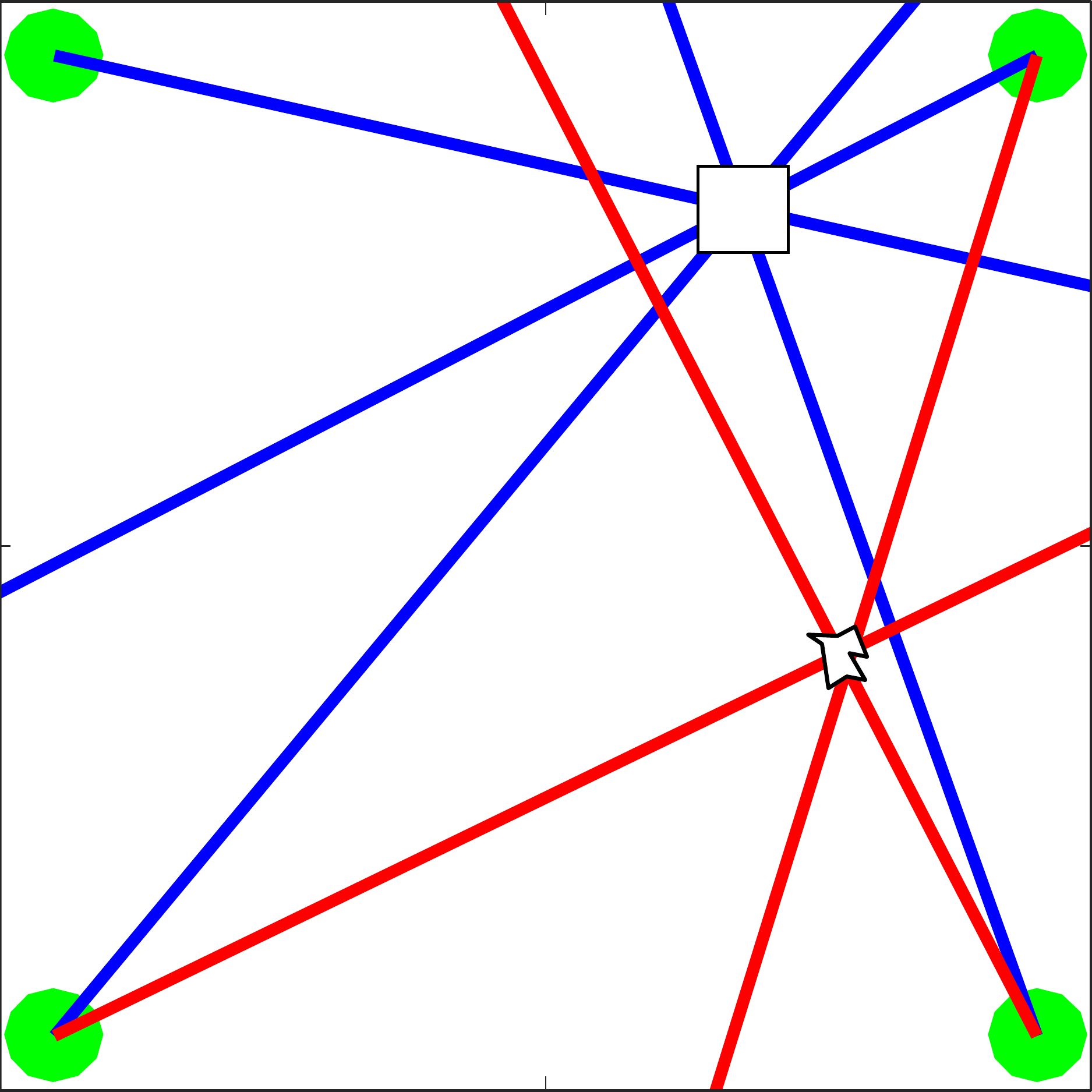}
	\caption{All base stations receive a LOS component from the source. In addition, all base stations except for the one on the top left corner, receive a NLOS component which resulted from bouncing on a reflector. The left figure plots the LOS and NLOS paths. The right figure plots the bearing lines at all base stations. In blue and red the LOS and NLOS components, respectively.}
	\label{fig:scenario_simple}
\end{figure}

\begin{figure}
	\begin{tikzpicture}
		\begin{axis}[
			xlabel={$w^2$},
			ylabel={Probability of sub-meter accuracy},
			xmin=0, xmax=5,
			ymin=0, ymax=1.05,
			legend entries={$\textrm{SNR}=\SI{0}{\decibel}$,$\textrm{SNR}=\SI{10}{\decibel}$,$\textrm{SNR}=\SI{20}{\decibel}$},
			legend pos=north west,
			legend cell align=left,
			cycle list name=myCycleList,
		]
		\addplot+[mark=none] table[
			x=wSquare,
			y=0dB,
		] {./Data/data_WeightSelection.dat};
		\addplot+[mark=none] table[
			x=wSquare,
			y=10dB,
		] {./Data/data_WeightSelection.dat};
		\addplot+[mark=none] table[
			x=wSquare,
			y=20dB,
		] {./Data/data_WeightSelection.dat};
		\end{axis}
	\end{tikzpicture}
	\caption{Probability of sub-meter accuracy versus the choice of the weight for the scenario in Fig.~\ref{fig:scenario_simple}. Probability estimated by Monte Carlo simulation where the random parameters are the signal strengths and phases.}
	\label{fig:weight_analysis}
\end{figure}
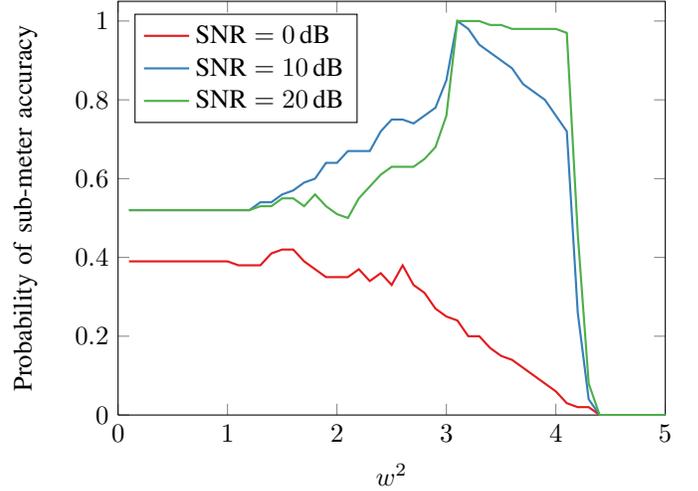
To illustrate Theorem~\ref{thrm:weight}, we synthesize a set of snapshots \eqref{eq:signal_model} according to the scenario plotted in Fig.~\ref{fig:scenario_simple} and ignore any time delay information. The source is positioned at [\SI{18}{\meter},\SI{31}{\meter}] and a reflector is positioned at [\SI{25}{\meter},\SI{-7}{\meter}].  As visualized in Fig.~\ref{fig:scenario_simple}, all base stations receive a LOS component, and except for the top left base station, they also receive a NLOS component bounced from the reflector.
From Fig.~\ref{fig:scenario_simple}, it is apparent that the source location is consistent with 4 paths, the reflector is consistent with 3 paths, and all other locations in that area are consistent with 2 paths or less. We hypothesize that for a sufficiently fine grid, and for a sufficiently high SNR, the probability of recovering the correct source location will be high if the weight is picked according to Theorem~\ref{thrm:weight}. The $\textrm{SNR}_l$ for the snapshots \eqref{eq:signal_model} is defined as $\textrm{SNR}_l={S_l\E\{|\bar{\alpha}_l|^2\}}/{\E\{\|\bar{\mathbf{n}}_l\|^2\}}={S_l\E\{|\bar{\alpha}_l|^2\}}/{\sigma^2}$ and is equal for all base stations $\textrm{SNR}_l=\textrm{SNR}$.
A Monte Carlo simulation is performed where at each run the signal strengths and phases of all multipath components are randomized according to Rayleigh and uniform distributions, respectively. The location of the source is estimated by running Algorithm~\ref{alg:grid_refinement}, wherein the solution to optimization problem \eqref{opt:problem} is obtained by the solver Mosek \cite{Mosek2016}.
Fig.~\ref{fig:weight_analysis} plots the empirical probability that the localization error is smaller than \SI{1}{\meter} as a function of $w^2$. 
According to Theorem~\ref{thrm:weight}, a sufficient condition for recovering the location of the source is that the square of the weight satisfies $L-1\leq w^2 \leq L$. The figure shows that, in this case, for $L=4$, the range of values $w^2\in[3,4]$ yields the correct source location with sub-meter accuracy with probability close to 1 for a sufficient high SNR.

\subsection{Localization Performance in Realistic Multipath Channel}

In this section, DiSouL is compared to indirect localization techniques: 
\begin{itemize}
\item SR-LS \cite{beck2008exact}, using TOA obtained by the time delay estimator of Section~\ref{sec:TOA}. 
\item IV \cite{douganccay2004passive}, using AOA information, obtained by applying beamforming \cite{godara1997application} on the snapshots \eqref{eq:signal_model} and selecting the angle associated with the strongest peak. 
\item The Stansfield estimator \cite{gavish1992performance}, using  hybrid TOA-AOA.
\item DPD \cite{weiss2004direct}, a direct localization hybrid TOA-AOA technique, operating directly on the received signals \eqref{eq:signal_model_preMF}--\eqref{eq:signal_model_NLOS}.
\end{itemize}
More sophisticated techniques for estimating AOAs such as MUSIC \cite{schmidt1986multiple}
are not applicable because they require multiple snapshots and break down in the presence of multiple correlated arrivals such as is the case of multipath. A high precision alternative to beamforming is $\ell_1$-SVD \cite{malioutov2005sparse}. However, we have observed in our numerical results that $\ell_1$-SVD performs similar to beamforming due to the fact that the AOA estimation errors are caused by peak ambiguities and not the lack of angular resolution. Thus, errors happen mostly when the LOS component is attenuated or blocked by obstacles.

The source emits a Gaussian pulse $s(t)$ at \SI{7}{\giga\hertz} carrier frequency. We simulate the received signal at each antenna after down-conversion to baseband and sampling. An oversampling factor of 3 is used. It is a assumed a half power bandwidth of $B = \SI{30}{\mega\hertz}$ and ${E}/{N_0}=\SI{20}{\decibel}$, where $E=\E|\alpha_l|^2$ is the energy of the received LOS component before sampling \eqref{eq:signal_model_preMF:LOS} (same energy for all $l$) and $N_0=\sigma^2$ is the noise spectral density at each antenna.
All parameters in the received signals \eqref{eq:signal_model_preMF} are generated according to the statistical indoor multipath channel in \cite{spencer2000modeling}. The configuration is that of the Clyde building: cluster decay rate is \SI{34}{\nano\second}, ray decay rate is \SI{29}{\nano\second}, cluster arrival rate is {1}/{\SI{17}{\nano\second}}, ray arrival rate is {1}/{\SI{5}{\nano\second}} and angular variance is \SI{26}{\degree}. On the average, at every base station, 99.9\% of the energy in the snapshot \eqref{eq:signal_model} is contained in 8 discrete multipath arrivals, which, in general, have closely spaced AOAs.

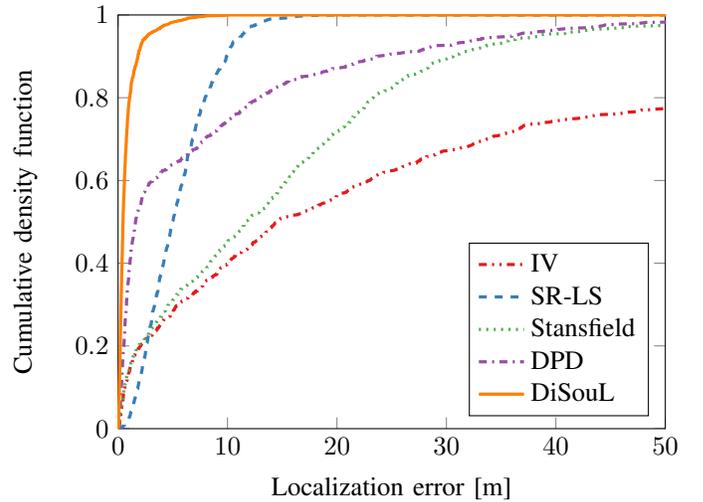
\begin{figure}
	\begin{tikzpicture}
	\begin{axis}[
	xlabel={Localization error [\si{\meter}]},
	ylabel={Cumulative density function},
	xmin=0, xmax=50,
	ymin=0, ymax=1,
	legend entries={IV,SR-LS,Stansfield,DPD,DiSouL},
	legend pos=south east,
	legend cell align=left,
	cycle list name=myCycleList,
	]
	\addplot+[mark=none,dashdotdotted,very thick] table[
	x=IV,
	y=CDF,
	] {./Data/data_CDF.dat};
	\addplot+[mark=none,dashed,very thick] table[
	x=SR-LS,
	y=CDF,
	] {./Data/data_CDF.dat};
	\addplot+[mark=none,dotted,very thick] table[
	x=Stansfield,
	y=CDF,
	] {./Data/data_CDF.dat};
	\addplot+[mark=none,dashdotted,very thick] table[
	x=DPD,
	y=CDF,
	] {./Data/data_CDF.dat};
	\addplot+[mark=none,very thick] table[
	x=DiSouL,
	y=CDF,
	] {./Data/data_CDF.dat};
	\end{axis}
	\end{tikzpicture}
	\caption{Cumulative density function of the localization error for ${E}/{N_0}=\SI{20}{\decibel}$ and $B=\SI{30}{\mega\hertz}$.}
	\label{fig:CDF}
\end{figure}
Fig.~\ref{fig:CDF} plots the cumulative density function of the localization error. Clearly, DiSouL achieves high precision accuracy with high probability, followed by DPD and the two-step approaches. To gain a more in-depth understanding, we will focus on the performance of the estimators at sub-meter errors, as a function of ${E}/{N_0}$, bandwidth, number of antennas, channel properties, and calibration errors.

\begin{figure}
	\begin{tikzpicture}
	\begin{axis}[
		xlabel={${E}/{N_0}$ [\si{\decibel}]},
		ylabel={Probability of sub-meter accuracy},
		xmin=-5, xmax=35,
		ymin=0, ymax=1,
		legend entries={IV,SR-LS,Stansfield,DPD,DiSouL},
		legend pos=north west,
		legend cell align=left,
		cycle list name=myCycleList,
		legend style={font=\small},
	]
	\addplot table[
		x=SNR,
		y=IV,
	] {./Data/data_ProbVsSNR.dat};
	\addplot table[
		x=SNR,
		y=SR-LS,
	] {./Data/data_ProbVsSNR.dat};
	\addplot table[
		x=SNR,
		y=Stansfield,
	] {./Data/data_ProbVsSNR.dat};
	\addplot table[
		x=SNR,
		y=DPD,
	] {./Data/data_ProbVsSNR.dat};
	\addplot table[
		x=SNR,
		y=DiSouL,
	] {./Data/data_ProbVsSNR.dat};
	\end{axis}
	\end{tikzpicture}
	\caption{Probability of sub-meter precision vs.\ ${E}/{N_0}$ for $B=\SI{30}{\mega\hertz}$.}
	\label{fig:prob1m}
\end{figure}
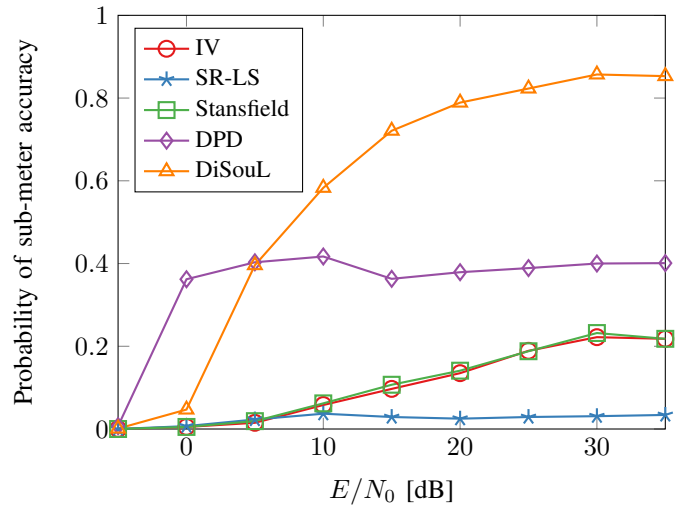
In Fig.~\ref{fig:prob1m}, the probability of sub-meter precision is shown as a function of ${E}/{N_0}$. 
Note that DiSouL outperforms all other techniques for most ${E}/{N_0}$ values. The TOA-based  SR-LS performs poorly due to the positive bias of the TOA estimates. The AOA-based estimators can slightly improve on this performance, but are still worse than both direct localization approaches. 
As ${E}/{N_0}$ increases, we sample the snapshots at the time of crossing a threshold rather than at the peak (see Section~\ref{sec:timing}), which reduces the amount of NLOS multipath components that are included into the snapshots, but the resulting ratio between LOS energy and noise is more or less independent of ${E}/{N_0}$. Thus, the benefit of increased ${E}/{N_0}$ is that we detect the signals sooner, thus diminishing the number of NLOS components making into the snapshots. However, this is is limited by the the resolution of the matched filter, and DiSouL's probability of sub-meter accuracy saturates over \SI{30}{\decibel} because of the limited time resolution of a digital matched filter.

\begin{figure}
	\begin{tikzpicture}
	\begin{axis}[
		xlabel={Bandwidth [\si{\mega\hertz}]},
		ylabel={Probability of sub-meter accuracy},
		xmin=10, xmax=100,
		ymin=0, ymax=1,
		legend entries={IV,SR-LS,Stansfield,DPD,DiSouL},
		legend pos=north east,
		legend cell align=left,
		legend style={at={(.97,.82)}},
		cycle list name=myCycleList,
		legend columns={2},
		legend style={font=\small},
	]
	\addplot table[
		x=BW,
		y=IV,
	] {./Data/data_ProbVsBandwidth.dat};
	\addplot table[
		x=BW,
		y=SR-LS,
	] {./Data/data_ProbVsBandwidth.dat};
	\addplot table[
		x=BW,
		y=Stansfield,
	] {./Data/data_ProbVsBandwidth.dat};
	\addplot table[
		x=BW,
		y=DPD,
	] {./Data/data_ProbVsBandwidth.dat};
	\addplot table[
		x=BW,
		y=DiSouL,
	] {./Data/data_ProbVsBandwidth.dat};
	\end{axis}
	\end{tikzpicture}
	\caption{Probability of sub-meter precision vs.\ bandwidth for ${E}/{N_0}=\SI{20}{\decibel}$.}
	\label{fig:BWprobab1m}
\end{figure}
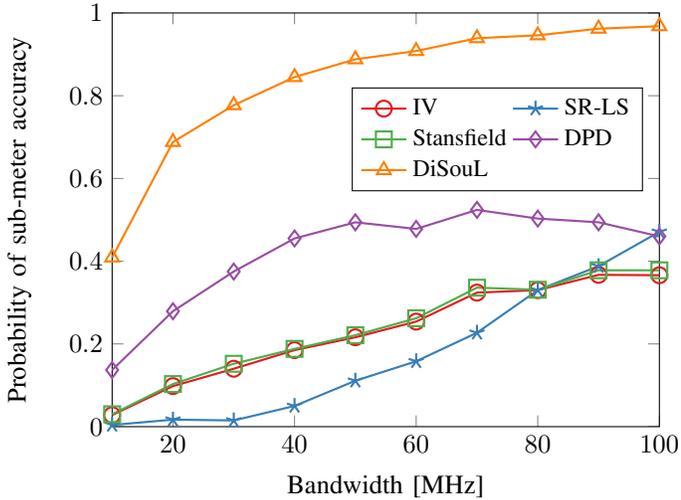
Fig.~\ref{fig:BWprobab1m} plots the probability of sub-meter accuracy versus signal bandwidth. All techniques benefit from an increase of bandwidth. On the one hand, it is well known that a larger bandwidth results in better TOA estimates. On the other hand, since the pulse width is inversely proportional to the bandwidth, a larger bandwidth results in a shorter pulse. Hence, fewer NLOS multipath components are included into the snapshots \eqref{eq:signal_model}, thus, decreasing the risk of errors in the AOA estimation.

\begin{figure}
	\begin{tikzpicture}
	\begin{axis}[
	xlabel={Number of antennas},
	ylabel={Probability of sub-meter accuracy},
	xmin=10, xmax=200,
	ymin=0, ymax=1,
	legend entries={IV,SR-LS,Stansfield,DPD,DiSouL},
	legend pos=north west,
	legend cell align=left,
	cycle list name=myCycleList,
	legend style={font=\footnotesize},
	]
	\addplot table[
	x=antennas,
	y=IV,
	] {./Data/data_numAntennas.dat};
	\addplot table[
	x=antennas,
	y=SR-LS,
	] {./Data/data_numAntennas.dat};
	\addplot table[
	x=antennas,
	y=Stansfield,
	] {./Data/data_numAntennas.dat};
	\addplot table[
	x=antennas,
	y=DPD,
	] {./Data/data_numAntennas.dat};
	\addplot table[
	x=antennas,
	y=DiSouL,
	] {./Data/data_numAntennas.dat};
	\end{axis}
	\end{tikzpicture}
	\caption{Probability of sub-meter precision vs.\ number of antennas at each base station for ${E}/{N_0}=\SI{20}{\decibel}$ and $B=\SI{30}{\mega\hertz}$.}
	\label{fig:numAntennas}
\end{figure}
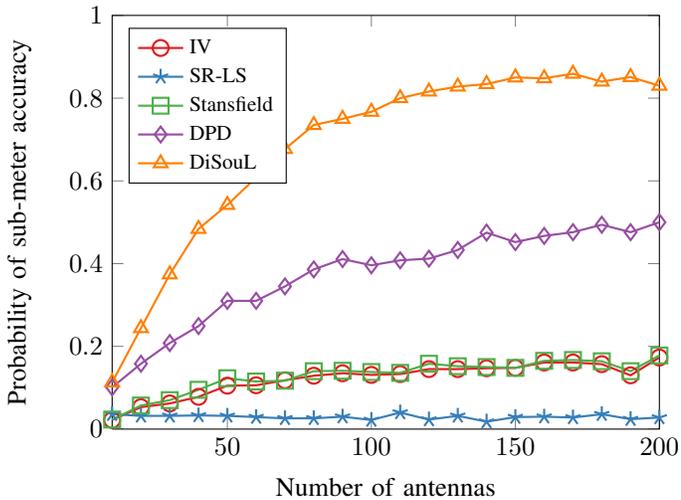
Fig.~\ref{fig:numAntennas} evaluates the probability of sub-meter accuracy versus the number of antennas in each base station. The size of the array in each base station grows with the number of antennas in order to maintain a constant average inter-antenna spacing. Due to the array size increase, the angular resolution at each base station also improves, which in turn allows DiSouL to better resolve multipath arrivals. On the contrary, the probability of sub-meter accuracy for the indirect techniques remains approximately the same. In particular, because SR-LS is purely TOA-based, the improvement in angular resolution has no impact. The other two indirect techniques, IV and Stansfield, improve very little because most of their errors are due to selection of the wrong path as LOS.

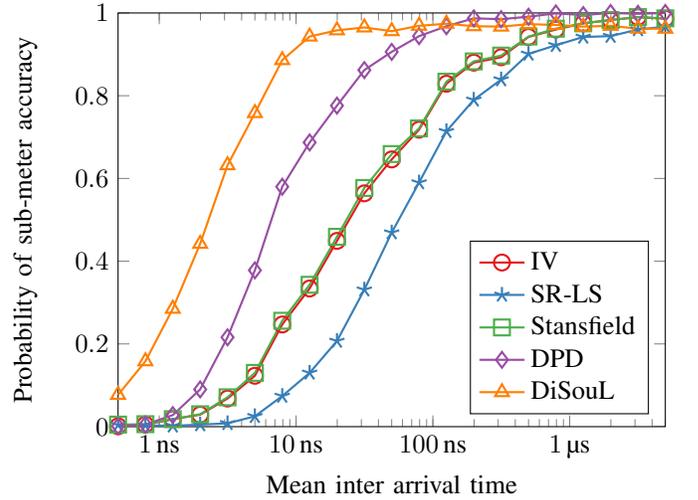
\begin{figure}
	\begin{tikzpicture}
	\begin{semilogxaxis}[
		xlabel={Mean inter arrival time},
		ylabel={Probability of sub-meter accuracy},
		xmin=5e-1, xmax=5e3,
		ymin=0, ymax=1,
		legend entries={IV,SR-LS,Stansfield,DPD,DiSouL},
		legend pos=south east,
		legend cell align=left,
		xtick={1,10,100,1e3,1e4},
		xticklabels={\SI{1}{\nano\second},\SI{10}{\nano\second},\SI{100}{\nano\second},\SI{1}{\micro\second},\SI{10}{\micro\second}},
		cycle list name=myCycleList,
	]
	\addplot table[
		x=time,
		y=IV,
	] {./Data/data_ArrivalTime.dat};
	\addplot table[
		x=time,
		y=SR-LS,
	] {./Data/data_ArrivalTime.dat};
	\addplot table[
		x=time,
		y=Stansfield,
	] {./Data/data_ArrivalTime.dat};
	\addplot table[
		x=time,
		y=DPD,
	] {./Data/data_ArrivalTime.dat};
	\addplot table[
		x=time,
		y=DiSouL,
	] {./Data/data_ArrivalTime.dat};
	\end{semilogxaxis}
	\end{tikzpicture}
	\caption{Probability of sub-meter precision vs.\ ray mean arrival time for ${E}/{N_0}=\SI{20}{\decibel}$ and $B=\SI{30}{\mega\hertz}$.}
	\label{fig:arrivalTime}
\end{figure}

In Fig.~\ref{fig:arrivalTime}, we tune some of the channel parameters controlling the rate of arrivals.  In the statistical multipath channel model of \cite{spencer2000modeling}, the times of arrival of the NLOS components are modeled by two parameters: the cluster arrival rate $\Lambda$ and the ray arrival rate $\lambda$. The measured values for the Clyde building of these two parameters were ${1}/{\Lambda}=\SI{17}{\nano\second}$ and ${1}/{\Lambda}=\SI{5}{\nano\second}$. In order to study the localization accuracy as a function of the ray arrival time, in Fig.~\ref{fig:arrivalTime}, ${1}/{\Lambda}$ is varied between \SI{5}{\pico\second} and \SI{5}{\micro\second} while $\Lambda=\frac{5}{17}\lambda$. As the ray inter arrival time increases, the multipath channel becomes less dense. For very high inter-arrival times, the channel can be considered almost pure LOS, and as expected all techniques improve their localization accuracy.

\begin{figure}
	\begin{tikzpicture}
	\begin{axis}[
		xlabel={Calibration error interval [\si{\degree}]},
		ylabel={Probability of sub-meter accuracy},
		xmin=0, xmax=120,
		ymin=0, ymax=1,
		legend pos=north east,
		legend cell align=left,
		legend style={text width=7em, at={(.95,-0.2)}},
		cycle list name=myCycleList,
		legend columns={2},
	]
	\addlegendimage{no markers,lines-1,thick}
	\addlegendentry{${E}/{N_0}=\SI{10}{\decibel}$  }
	\addlegendimage{only marks,mark=o,mark size=3}
	\addlegendentry{$\text{ B}=\SI{10}{\mega\hertz}$}
	\addlegendimage{no markers,lines-2,thick}
	\addlegendentry{${E}/{N_0}=\SI{20}{\decibel}$}
	\addlegendimage{only marks,mark=star,mark size=3}
	\addlegendentry{$\text{ B}=\SI{30}{\mega\hertz}$}
	\addlegendimage{no markers,lines-3,thick}
	\addlegendentry{${E}/{N_0}=\SI{30}{\decibel}$}
	\addlegendimage{only marks,mark=square,mark size=3}
	\addlegendentry{$\text{ B}=\SI{100}{\mega\hertz}$}
	\addplot+[lines-1,mark=o] table[
		x=calibration,
		y=1,
	] {./Data/data_Calibration.dat};
	\addplot+[lines-2,mark=o] table[
		x=calibration,
		y=2,
	] {./Data/data_Calibration.dat};
	\addplot+[lines-3,mark=o] table[
		x=calibration,
		y=3,
	] {./Data/data_Calibration.dat};
	\addplot+[lines-1,mark=star,thick] table[
		x=calibration,
		y=4,
	] {./Data/data_Calibration.dat};
	\addplot+[lines-2,mark=star,thick] table[
		x=calibration,
		y=5,
	] {./Data/data_Calibration.dat};
	\addplot+[lines-3,mark=star,thick] table[
		x=calibration,
		y=6,
	] {./Data/data_Calibration.dat};
	\addplot+[lines-1,mark=square] table[
		x=calibration,
		y=7,
	] {./Data/data_Calibration.dat};
	\addplot+[lines-2,mark=square] table[
		x=calibration,
		y=8,
	] {./Data/data_Calibration.dat};
	\addplot+[lines-3,mark=square] table[
		x=calibration,
		y=9,
	] {./Data/data_Calibration.dat};
	\end{axis}
	\end{tikzpicture}
	\caption{Comparison of DiSouL's localization accuracy in the presence of calibration errors for ${E}/{N_0}=\SI{20}{\decibel}$ and $B=\SI{30}{\mega\hertz}$. The calibration errors are modelled as additive phase noise distributed according to a uniform distribution $\mathcal{U}(-{I}/{2},{I}/{2})$ where $I$ is the calibration error interval, and are independent across antennas.}
	\label{fig:calibration}
\end{figure}
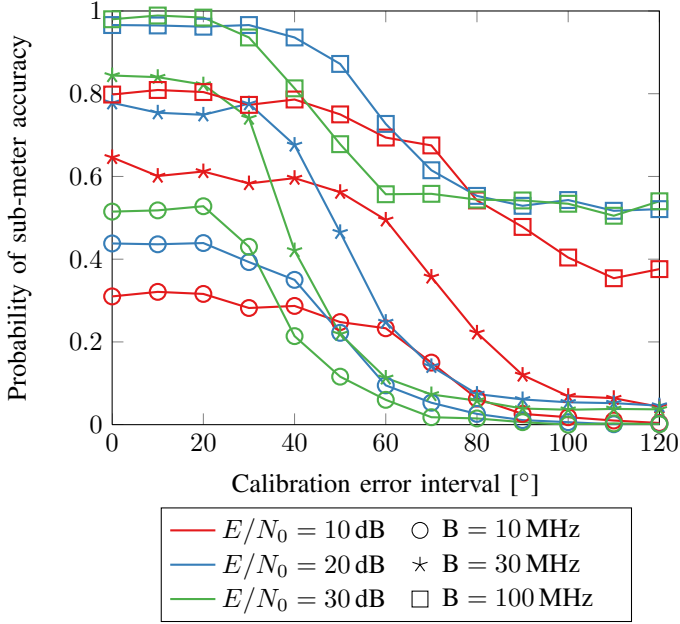
In Fig.~\ref{fig:calibration}, we analyze the effect of calibration errors in all arrays. The calibration errors are modelled as additive phase noise distributed according to a uniform distribution $\mathcal{U}(-{I}/{2},{I}/{2})$ where $I$ is the calibration error interval, and are independent across antennas.
For all tested values of bandwidth and ${E}/{N_0}$, DiSouL is robust to calibration errors intervals smaller than \SI{20}{\degree}. 
At \SI{100}{\mega\hertz}, even when the calibration errors are very large (\SI{120}{\degree}), DiSouL's probability of sub-meter accuracy does not drop to zero thanks to the assistance provided by the TOA estimates.

Lastly, Table~\ref{tab:executionTime} plots the execution times of all techniques in a regular \SI{3.6}{\giga\hertz} desktop computer. Due to the joint processing of the data at all base stations, the execution times of the direct techniques are much larger than those of the indirect techniques. In particular, DiSouL is substantially more computationally intensive than the other techniques because it needs to solve a relatively large optimization problem multiple times. Instead of using an off-the-shelf solver \cite{Mosek2016}, it may be worth developing an algorithm tailored to the specific structure of DiSouL's optimization problem.
\begin{table}
	\centering
	\caption{Average execution times of the multiple localization techniques.}
	\begin{tabular}{llllll}
		\toprule
		& 	IV    & SR-LS & Stansfield & DPD & DiSouL \\
		\midrule
		\parbox[c][6ex][c]{.13\columnwidth}{Average execution time} & \parbox[c][6ex][c]{.11\columnwidth}{\SI{0.38}{\milli\second}} & \parbox[c][6ex][c]{.11\columnwidth}{\SI{7.0}{\milli\second}} & \parbox[c][6ex][c]{.11\columnwidth}{\SI{0.22}{\milli\second}} & \parbox[c][6ex][c]{.11\columnwidth}{\SI{0.13}{\second}} & \parbox[c][6ex][c]{.11\columnwidth}{\SI{5.7}{\second}} \\
		\bottomrule
	\end{tabular}
	\label{tab:executionTime}
\end{table}

\section{Conclusions}

This paper tackled the problem of narrowband localization in the presence of multipath, through a direct localization approach in a massive MIMO setting. We propose an original compressive sensing approach for the localization of sources emitting known narrow-band signals. Due to the high angular resolution of massive arrays, it is possible to estimate the AOAs of the multipath components. By jointly processing snapshots of several widely distributed arrays, we are able to estimate the source location precisely without explicitly estimating the LOS AOAs, and therefore, avoiding the challenging data association problem. The proposed technique, called DiSouL, achieves sub-meter localization with high probability in dense multipath environments with narrow-band signals. DiSouL requires no statistical channel knowledge except for the noise variance and therefore it is suitable for any multipath environment. Coarse TOA estimates at each array are used to reduce the execution time and enhance the localization accuracy. Numerical simulations have shown that DiSouL is also very robust to synchronization errors, even though a better understanding is needed. The large gain in accuracy comes with higher computational complexity compared to previous existing techniques. 

\appendices

\section{Proof of Lemma~\ref{lemma:proof:minLpaths}} \label{proof:minLpaths}

We aim to prove that under A2) and A3), if $w > \sqrt{L-1}$, then any estimated location is consistent with $L$ paths (in the sense of Definition~\ref{def:consistency}). Here the point is that it is more costly (in terms of the objective function) to explain an observation as $L$ NLOS angles, than as one position with $L$ associated LOS angles, and this is so exactly when $w>\sqrt{L-1}$.

Let $\mathbf{X}$ and $\{\mathbf{y}\}_{l=1}^L$ be a solution from \eqref{opt:problem} with cost $C_1$, and let $\boldsymbol{\bm{\pi}}_1$ be an estimated location (i.e., $\sqrt{\sum_{l=1}^{L} \left| x_{1l} \right|^2}\neq 0$). Then $\hat{\mathbf{z}}_l$ can be expressed as
\begin{equation} \label{proof:minLpaths_1}
\hat{\mathbf{z}}_l = 
x_{1l} \mathbf{a}_l\left(\theta_l(\boldsymbol{\bm{\pi}}_1)\right)
+ y_{1 l} \mathbf{a}_l\left(\vartheta_{1l}\right)
+ \mathbf{e}_l, \quad l=1,\ldots,L,
\end{equation}
where 
$\vartheta_{1l}=\theta_l(\boldsymbol{\bm{\pi}}_1)$
and $\mathbf{e}_l$ is placeholder of all the terms in the reconstruction $\hat{\mathbf{z}}_l$ that are not related to the location $\boldsymbol{\bm{\pi}}_1$ or angle $\vartheta_{1l}$
\begin{equation} \label{proof:minLpaths_2}
\mathbf{e}_l = 
\sum_{q>1} x_{ql} \mathbf{a}_l\left(\theta_l(\boldsymbol{\bm{\pi}}_q)\right)
+\sum_{m>1}  y_{m l} \mathbf{a}_l\left(\vartheta_{ml}\right) \quad{l=1,\ldots,L} .
\end{equation}
Now, if $\left\|\mathbf{x}_{1}\right\|_0 = L$, then $\boldsymbol{\bm{\pi}}_1$ is consistent with $L$ paths in the sense of Definition~\ref{def:consistency} due to Assumption {A3)}. Here, $\mathbf{x}_{1}=[x_{11}\cdots x_{1L}]^\text{T}$ and $\|\cdot\|_0$ is the $\ell_0$-norm which counts the number of estimated elements.
Hence, we must prove that $w > \sqrt{L-1}$ implies $\left\|\mathbf{x}_{1}\right\|_0 = L$. The proof is by contradiction.

Assume that
\begin{equation} \label{proof:minLpaths_3.5}
\left\|\mathbf{x}_{1}\right\|_0 <L.
\end{equation}
This means that the position is consistent with less than $L$ paths. Now we have another competing reconstruction $\mathbf{X}'$, $\{\mathbf{y}'\}_{l=1}^L$ and
\begin{equation} \label{proof:minLpaths_4}
\hat{\mathbf{z}}_l = 
\underbrace{\left(x_{1l} +y_{1l}\right)}_{\doteq y'_{1l}} \mathbf{a}_l\left(\vartheta_{1l}\right)
+ \mathbf{e}_l \quad\text{for }l=1,\ldots,L,
\end{equation}
with $\mathbf{x}'_{1}=0$ and cost $C_2$. Since $\mathbf{X}$ and $\{\mathbf{y}\}_{l=1}^L$ are optimal, $C_1\leq C_2$.
Consider now the two assignments \eqref{proof:minLpaths_1} and \eqref{proof:minLpaths_4}. Ignoring any common coefficients, the cost of \eqref{proof:minLpaths_1} is
\begin{equation} \label{proof:minLpaths_5}
C_1= w \sqrt{\sum_{l=1}^{L} \left|x_{1l}\right|^2} + \sum_{l=1}^{L}\left|y_{1l}\right|
\end{equation}
whereas the cost of \eqref{proof:minLpaths_4} is 
\begin{equation} \label{proof:minLpaths_6}
C_2=  \sum_{l=1}^{L}\left| y'_{1 l} \right| = \sum_{l=1}^{L}\left|x_{1l}+y_{1l}\right| \le \sum_{l=1}^{L}\left|x_{1l}\right| +\sum_{l=1}^{L}\left|y_{1l}\right|.
\end{equation}
Since $C_1 \leq C_2$,
\begin{equation} \label{proof:minLpaths_7}
w \sqrt{\sum_{l=1}^{L}\left|x_{1l}\right|^2} \leq
\sum_{l=1}^{L}\left|x_{1l}\right|.
\end{equation}
Define the vector function $\mathbbm{1}_{\mathbf{x}}$ whose $l$-th entry is 1 if $x_{1l}\neq0$, and 0 otherwise, and denote by $\tilde{\mathbf{x}}$ the element-wise absolute value of $\mathbf{x}_{1}$, i.e., $\tilde{{x}}_l=|{x}_{1l}|$.  Then, $\Vert \mathbf{x}_{1} \Vert_1 = \mathbbm{1}_{\mathbf{x}}^\text{T} \tilde{\mathbf{x}}$, so from the Cauchy-Schwarz inequality, it follows immediately that
\begin{equation} \label{proof:minLpaths_7bis}
	\Vert \mathbf{x}_{1} \Vert_1 \le \sqrt{\Vert \mathbf{x}_{1} \Vert_0} \Vert \mathbf{x}_{1} \Vert_2	.
\end{equation}
Putting everything together, we find the following contradiction
\begin{equation} \label{proof:minLpaths_8}
\begin{split}
	w\Vert\mathbf{x}_{1} \Vert_2 &\overset{\eqref{proof:minLpaths_7}}{\leq} \Vert \mathbf{x}_{1} \Vert_1 \overset{\eqref{proof:minLpaths_7bis}}{\leq} \sqrt{\Vert \mathbf{x}_{1} \Vert_0} \Vert \mathbf{x}_{1} \Vert_2 \\
	&\overset{\eqref{proof:minLpaths_3.5}}{\leq} \sqrt{L-1}\Vert \mathbf{x}_{1} \Vert_2 \overset{(a)}{<} w \Vert\mathbf{x}_{1}\Vert_2,
\end{split}
\end{equation}
where $(a)$ is due to the fact that $w > \sqrt{L-1}$. Hence, $w > \sqrt{L-1}$ implies $\left\|\mathbf{x}_{1}\right\|_0 =L$.

\section{Proof of Lemma~\ref{lemma:min1source}} \label{proof:min1source}

If no location were found, for each possible location just enough ``mass'' from each NLOS detected observation could be moved, in a certain way, over to LOS; the new cost cannot exceed the nominal cost if $w<\sqrt{L}$. 

The proof is by contradiction. Assume that $w < \sqrt{L}$ and that there is no estimated location output by problem \eqref{opt:problem}, so that $x_{ql}=0$, $\forall q, l$. Then, 
\begin{equation} \label{proof:min1source_1}
\hat{\mathbf{z}}_l =		
\sum_{m} y_{ml} \mathbf{a}_l\left(\vartheta_{ml}\right) \quad{l=1,\ldots,L}.
\end{equation}
Assume without loss of generality that $\vartheta_{1l}=\theta_l(\mathbf{p})$.
By Assumption {A3)}, $\theta_l(\mathbf{p})\in\hat{\Theta}_l$, so that $y_{1l} \neq 0$, which leads to the following decomposition
\begin{equation} \label{proof:min1source_5}
\hat{\mathbf{z}}_l =
y_{ 1 l} \mathbf{a}_l\left(\vartheta_{1l}\right) +	
\sum_{m>1} y_{m l} \mathbf{a}_l\left(\vartheta\right) \quad{l=1,\ldots,L}.
\end{equation}
We consider a competing decomposition $\mathbf{X}'$, $\{\mathbf{y}'\}_{l=1}^L$, for which $x'_{ql}\neq 0$ for some $q,l$. In particular, $\boldsymbol{\bm{\pi}}_1=\mathbf{p}$; then 
\begin{multline} \label{proof:min1source_6}
	\hat{\mathbf{z}}_l =
	x'_{1l} \mathbf{a}_l\left(\theta_l(\boldsymbol{\bm{\pi}}_1)\right) +
	\underbrace{\left( y_{1 l} -x'_{1l}\right)}_{\doteq y'_{1l}} \mathbf{a}_l\left(\vartheta_{1l}\right)
	+\\+
	\sum_{m>1} y'_{ml} \mathbf{a}_l\left(\vartheta_{ml}\right) \qquad{l=1,\ldots,L},
\end{multline}
where $y'_{ml}=y_{ml}$ for all $l$ and $m>1$.
Ignoring common terms, we can associate a cost $C_1$ and $C_2$ with \eqref{proof:min1source_5} and \eqref{proof:min1source_6}, respectively, where
\begin{align} \label{proof:min1source_8}
C_1 &= \sum_{l=1}^{L}\left|y_{1 l} \right| \\
C_2 &= w\sqrt{\sum_{l=1}^{L}\left|x'_{1l}\right|^2} +
\sum_{l=1}^{L}\left| y_{1l}-x'_{1l}\right|.
\end{align}
If we select $x'_{1l}$ such that $|x'_{1l}| = \min_{l} |y_{1l}|$ and $\angle x'_{1l} = \angle y_{1l}$, and utilize the fact that $C_1\le C_2$, we have 
\begin{equation}
\sum_{l=1}^{L}\left|y_{1l} \right|  \le  w \sqrt{L} \min_{l} |y_{1l}| 
+\sum_{l=1}^{L}| y_{1l}| - L \min_{l} |y_{1l}|,
\end{equation}
implying that $w \sqrt{L} -L \ge 0$, which contradicts $w < \sqrt{L}$.

\bibliographystyle{IEEEtran}
\bibliography{IEEEabrv,../references}

\begin{thebibliography}{10}
\providecommand{\url}[1]{#1}
\csname url@samestyle\endcsname
\providecommand{\newblock}{\relax}
\providecommand{\bibinfo}[2]{#2}
\providecommand{\BIBentrySTDinterwordspacing}{\spaceskip=0pt\relax}
\providecommand{\BIBentryALTinterwordstretchfactor}{4}
\providecommand{\BIBentryALTinterwordspacing}{\spaceskip=\fontdimen2\font plus
\BIBentryALTinterwordstretchfactor\fontdimen3\font minus
  \fontdimen4\font\relax}
\providecommand{\BIBforeignlanguage}[2]{{%
\expandafter\ifx\csname l@#1\endcsname\relax
\typeout{** WARNING: IEEEtran.bst: No hyphenation pattern has been}%
\typeout{** loaded for the language `#1'. Using the pattern for}%
\typeout{** the default language instead.}%
\else
\language=\csname l@#1\endcsname
\fi
#2}}
\providecommand{\BIBdecl}{\relax}
\BIBdecl

\bibitem{rappaport2013millimeter}
T.~S. Rappaport, S.~Sun, R.~Mayzus, H.~Zhao, Y.~Azar, K.~Wang, G.~N. Wong,
  J.~K. Schulz, M.~Samimi, and F.~Gutierrez, ``Millimeter wave mobile
  communications for {5G} cellular: It will work!'' \emph{IEEE Access}, vol.~1,
  pp. 335--349, 2013.

\bibitem{larsson2014massive}
E.~G. Larsson, O.~Edfors, F.~Tufvesson, and T.~Marzetta, ``Massive {MIMO} for
  next generation wireless systems,'' \emph{IEEE Communications Magazine},
  vol.~52, no.~2, pp. 186--195, 2014.

\bibitem{chin2014emerging}
W.~H. Chin, Z.~Fan, and R.~Haines, ``Emerging technologies and research
  challenges for {5G} wireless networks,'' \emph{IEEE Wireless Communications},
  vol.~21, no.~2, pp. 106--112, 2014.

\bibitem{guerra2015position}
A.~Guerra, F.~Guidi, and D.~Dardari, ``Position and orientation error bound for
  wideband massive antenna arrays,'' in \emph{IEEE International Conference on
  Communication Workshop}, 2015, pp. 853--858.

\bibitem{savic2015fingerprinting}
V.~Savic and E.~G. Larsson, ``Fingerprinting-based positioning in distributed
  massive {MIMO} systems,'' in \emph{IEEE 82nd Vehicular Technology
  Conference}, 2015, pp. 1--5.

\bibitem{guidi2016personal}
F.~Guidi, A.~Guerra, and D.~Dardari, ``Personal mobile radars with
  millimeter-wave massive arrays for indoor mapping,'' \emph{IEEE Transactions
  on Mobile Computing}, vol.~15, no.~6, pp. 1536--1233, 2016.

\bibitem{guidi2014millimeter}
------, ``Millimeter-wave massive arrays for indoor {SLAM},'' in \emph{IEEE
  International Conference on Communications Workshops}, 2014, pp. 114--120.

\bibitem{gavish1992performance}
M.~Gavish and A.~J. Weiss, ``Performance analysis of bearing-only target
  location algorithms,'' \emph{IEEE Transactions on Aerospace and Electronic
  Systems}, vol.~28, no.~3, pp. 817--828, 1992.

\bibitem{kaplan2001maximum}
L.~M. Kaplan, Q.~Le, and P.~Moln{\'a}r, ``Maximum likelihood methods for
  bearings-only target localization,'' in \emph{IEEE International Conference
  on Acoustics, Speech, and Signal Processing}, vol.~5, 2001, pp. 3001--3004.

\bibitem{douganccay2004passive}
K.~Do{\u{g}}an{\c{c}}ay, ``Passive emitter localization using weighted
  instrumental variables,'' \emph{Signal processing}, vol.~84, no.~3, pp.
  487--497, 2004.

\bibitem{azzouzi2011new}
S.~Azzouzi, M.~Cremer, U.~Dettmar, R.~Kronberger, and T.~Knie, ``New
  measurement results for the localization of {UHF} {RFID} transponders using
  an angle of arrival ({AOA}) approach,'' in \emph{IEEE International
  Conference on RFID}, 2011, pp. 91--97.

\bibitem{klukas1998line}
R.~Klukas and M.~Fattouche, ``Line-of-sight angle of arrival estimation in the
  outdoor multipath environment,'' \emph{IEEE Transactions on Vehicular
  Technology}, vol.~47, no.~1, pp. 342--351, 1998.

\bibitem{spencer2000modeling}
Q.~H. Spencer, B.~D. Jeffs, M.~Jensen \emph{et~al.}, ``Modeling the statistical
  time and angle of arrival characteristics of an indoor multipath channel,''
  \emph{IEEE Journal on Selected Areas in Communications}, vol.~18, no.~3, pp.
  347--360, 2000.

\bibitem{sen2013avoiding}
S.~Sen, J.~Lee, K.-H. Kim, and P.~Congdon, ``Avoiding multipath to revive
  inbuilding wifi localization,'' in \emph{11th Annual International Conference
  on Mobile Systems, Applications, and Services}.\hskip 1em plus 0.5em minus
  0.4em\relax ACM, 2013, pp. 249--262.

\bibitem{pattipati1992new}
K.~R. Pattipati, S.~Deb, Y.~Bar-Shalom, and R.~B. Washburn~Jr, ``A new
  relaxation algorithm and passive sensor data association,'' \emph{IEEE
  Transactions on Automatic Control}, vol.~37, no.~2, pp. 198--213, 1992.

\bibitem{weiss2004direct}
A.~J. Weiss, ``Direct position determination of narrowband radio frequency
  transmitters,'' \emph{IEEE Signal Processing Letters}, vol.~11, no.~5, pp.
  513--516, 2004.

\bibitem{wax1982location}
M.~Wax, T.-J. Shan, and T.~Kailath, ``Location and the spectral density
  estimation of multiple sources,'' DTIC Document, Tech. Rep., 1982.

\bibitem{wax1983optimum}
M.~Wax and T.~Kailath, ``Optimum localization of multiple sources by passive
  arrays,'' \emph{IEEE Transactions on Acoustics, Speech and Signal
  Processing}, vol.~31, no.~5, pp. 1210--1217, 1983.

\bibitem{wax1985decentralized}
------, ``Decentralized processing in sensor arrays,'' \emph{IEEE Transactions
  on Acoustics, Speech and Signal Processing}, vol.~33, no.~5, pp. 1123--1129,
  1985.

\bibitem{weiss2005direct}
A.~J. Weiss and A.~Amar, ``Direct position determination of multiple radio
  signals,'' \emph{EURASIP Journal on Applied Signal Processing}, vol. 2005,
  no.~1, pp. 37--49, 2005.

\bibitem{chen2002maximum}
J.~C. Chen, R.~E. Hudson, and K.~Yao, ``Maximum-likelihood source localization
  and unknown sensor location estimation for wideband signals in the
  near-field,'' \emph{IEEE Transactions on Signal Processing}, vol.~50, no.~8,
  pp. 1843--1854, 2002.

\bibitem{bialer2013maximum}
O.~Bialer, D.~Raphaeli, and A.~J. Weiss, ``Maximum-likelihood direct position
  estimation in dense multipath,'' \emph{IEEE Transactions on Vehicular
  Technology}, vol.~62, no.~5, pp. 2069--2079, 2013.

\bibitem{mobile2011cran}
``{C-RAN}: the road towards green {RAN},'' White Paper, {China Mobile}, October
  2011.

\bibitem{wu2012green}
J.~Wu, S.~Rangan, and H.~Zhang, \emph{Green communications: theoretical
  fundamentals, algorithms and applications}.\hskip 1em plus 0.5em minus
  0.4em\relax CRC Press, 2012.

\bibitem{garcia2014direct}
N.~Garcia, A.~M. Haimovich, J.~A. Dabin, M.~Coulon, and M.~Lops, ``Direct
  localization of emitters using widely spaced sensors in multipath
  environments,'' in \emph{48th Asilomar Conference on Signals, Systems and
  Computers}.\hskip 1em plus 0.5em minus 0.4em\relax IEEE, 2014, pp. 695--700.

\bibitem{malioutov2005sparse}
D.~Malioutov, M.~{\c{C}}etin, and A.~S. Willsky, ``A sparse signal
  reconstruction perspective for source localization with sensor arrays,''
  \emph{IEEE Transactions on Signal Processing}, vol.~53, no.~8, pp.
  3010--3022, 2005.

\bibitem{dardari2008threshold}
D.~Dardari, C.-C. Chong, and M.~Z. Win, ``Threshold-based time-of-arrival
  estimators in {UWB} dense multipath channels,'' \emph{IEEE Transactions on
  Communications}, vol.~56, no.~8, pp. 1366--1378, 2008.

\bibitem{cespedes1995methods}
I.~Cespedes, Y.~Huang, J.~Ophir, and S.~Spratt, ``Methods for estimation of
  subsample time delays of digitized echo signals,'' \emph{Ultrasonic imaging},
  vol.~17, no.~2, pp. 142--171, 1995.

\bibitem{tropp2006algorithms}
J.~A. Tropp, ``Algorithms for simultaneous sparse approximation. part {II}:
  Convex relaxation,'' \emph{Signal Processing}, vol.~86, no.~3, pp. 589--602,
  2006.

\bibitem{jacob2009group}
L.~Jacob, G.~Obozinski, and J.-P. Vert, ``Group lasso with overlap and graph
  lasso,'' in \emph{26th Annual International Conference on Machine
  Learning}.\hskip 1em plus 0.5em minus 0.4em\relax ACM, 2009, pp. 433--440.

\bibitem{lobo1998applications}
M.~S. Lobo, L.~Vandenberghe, S.~Boyd, and H.~Lebret, ``Applications of
  second-order cone programming,'' \emph{Linear algebra and its applications},
  vol. 284, no.~1, pp. 193--228, 1998.

\bibitem{hyder2010direction}
M.~M. Hyder and K.~Mahata, ``Direction-of-arrival estimation using a mixed
  $\ell_{2,0}$ norm approximation,'' \emph{IEEE Transactions on Signal
  Processing}, vol.~58, no.~9, pp. 4646--4655, 2010.

\bibitem{schechter1996handbook}
E.~Schechter, \emph{Handbook of Analysis and its Foundations}.\hskip 1em plus
  0.5em minus 0.4em\relax Academic Press, 1996.

\bibitem{ochiai2005collaborative}
H.~Ochiai, P.~Mitran, H.~V. Poor, and V.~Tarokh, ``Collaborative beamforming
  for distributed wireless ad hoc sensor networks,'' \emph{IEEE Transactions on
  Signal Processing}, vol.~53, no.~11, pp. 4110--4124, 2005.

\bibitem{std1979antennas}
\emph{Standard Test Procedures for Antennas}, IEEE Std. 149--1979, 1979.

\bibitem{Mosek2016}
{MOSEK ApS}, ``The {MOSEK} optimization toolbox for {MATLAB} manual, version
  7.1 (revision 51),'' \url{http://mosek.com}, (accessed on March 20, 2016).

\bibitem{beck2008exact}
A.~Beck, P.~Stoica, and J.~Li, ``Exact and approximate solutions of source
  localization problems,'' \emph{IEEE Transactions on Signal Processing},
  vol.~56, no.~5, pp. 1770--1778, 2008.

\bibitem{godara1997application}
L.~C. Godara, ``Application of antenna arrays to mobile communications.\ {II}.\
  {B}eam-forming and direction-of-arrival considerations,'' \emph{Proceedings
  of the IEEE}, vol.~85, no.~8, pp. 1195--1245, 1997.

\bibitem{schmidt1986multiple}
R.~O. Schmidt, ``Multiple emitter location and signal parameter estimation,''
  \emph{IEEE Transactions on Antennas and Propagation}, vol.~34, no.~3, pp.
  276--280, 1986.

\end{thebibliography}

\end{document}